%% file: BisLogCoa.tex
\documentclass{llncs}

\usepackage{amsmath,amssymb}
\usepackage{arrows}
\usepackage{color}
\usepackage{xcolor}
\usepackage{graphicx}
\usepackage{xspace}
\usepackage{array}
\usepackage{subfig}
\usepackage[linesnumbered,ruled]{algorithm2e}

\usepackage{mathrsfs}

\newcommand{\calC}{{\cal C}}
\newcommand{\calF}{{\cal F}}
\newcommand{\calI}{{\cal I}}

\newcommand{\calN}{{\cal N}}
\newcommand{\calM}{{\cal M}}
\newcommand{\calP}{{\cal P}}
\newcommand{\calS}{{\cal S}}
\newcommand{\calT}{{\cal T}}
\newcommand{\calV}{{\cal V}}
\newcommand{\calX}{{\cal X}}
\newcommand{\calY}{{\cal Y}}

\newcommand{\RED}[1]{\textcolor{red}{#1}}
\newcommand{\BROWN}[1]{\textcolor{brown}{#1}}

\newcommand{\MAGENTA}[1]{\textcolor{magenta}{#1}}

\newcounter{dgnot}
\newenvironment{dgnot}[1][]{\refstepcounter{dgnot}\par\medskip
   \noindent \textbf{\RED{NfD~\thedgnot.}  #1} \rmfamily}{\medskip}

\newcounter{mknot}
\newenvironment{mknot}[1][]{\refstepcounter{mknot}\par\medskip
   \noindent \textbf{\MAGENTA{NfM~\themknot.}  #1} \rmfamily}{\medskip}

\newcounter{vcnot}
\newenvironment{vcnot}[1][]{\refstepcounter{vcnot}\par\medskip
   \noindent \textbf{\BROWN{NfV~\thevcnot.}  #1} \rmfamily}{\medskip}

\newcommand{\pfend}{$ $\hfill $\bullet$}
\renewcommand{\pfend}{\qed}
\def\deriv{\noindent\hspace*{.20in}\vspace{0.1in}}

\def\hint#1#2{\newline#1\hspace*{.25in}\{$#2$\}\vspace{0.1in}\newline\hspace*{.20in}\vspace{0.1in}}
\def\nh#1{\newline#1\hspace*{.25in}\vspace{0.1in}\newline\hspace*{.20in}\vspace{0.1in}}

\newcommand{\SET}[1]{\{ #1 \}}
\newcommand{\ZET}[2]{\SET{ \, #1 \mid #2 \,}}
\newcommand{\pws}[1]{\mathcal{P} #1}
\newcommand{\fpws}[1]{\mathcal{P}_{\omega} #1}
\newcommand{\compose}{\mathop{\raisebox{0.5pt}{\scriptsize $\circ$}}}
\newcommand{\nats}{\mathbb{N}}
\renewcommand{\succ}{\mathtt{Succ}}

\newcommand{\closurename}[1]{\calC_{#1}}
\newcommand\cmbisim{ \, \cong_{\closurename{}}  }
\newcommand\gcmbisim{\, \sim \, }

\newcommand\altcmbisim{ \, \simeq_{\closurename{}} }

\newcommand{\closure}[2]{\closurename{#1} \mkern1mu #2}
\newcommand{\forwclosurename}{\stackrel{\rightarrow}{\closurename{}}}
\newcommand{\forwclosure}[1]{\forwclosurename\!\!#1}
\newcommand{\backclosurename}{\stackrel{\leftarrow}{\closurename{}}}
\newcommand{\backclosure}[1]{\backclosurename\!\!#1}

\newcommand{\interiorname}[1]{\calI_{#1}}
\newcommand{\interior}[2]{\interiorname{#1}\,#2}
\newcommand{\Set}{\textbf{Set}}
\newcommand{\fmorph}[3]{[\![ {#3} ]\!]^{#2}_{#1}}
\newcommand{\beq}[2]{\approx^{{#2}}_{{#1}}}
\newcommand\ebisim{ \, \simeq_{\eta} \, }

\newcommand{\SLCS}{\mathtt{SLCS}\xspace}
\newcommand{\msep}{\,\mid\,}
\newcommand{\form}{\Phi}
\newcommand{\lneg}{\neg}
\newcommand{\lfalse}{\bot}
\newcommand{\ltrue}{\top}
\newcommand{\lrcsname}{{\stackrel{\rightarrow}{\rho}}}
\newcommand{\lrcs}[2]{\mbox{\footnotesize $\lrcsname$}\, #1[#2]}
\newcommand{\lrcdname}{{\stackrel{\leftarrow}{\rho}}}
\newcommand{\lrcd}[2]{\mbox{\footnotesize $\lrcdname$}\, #1[#2]}
\newcommand{\lnear}{\calN}

\newcommand{\lsurr}{\calS}
\newcommand{\lprop}{\calP}

\newcommand{\model}{\calM}
\newcommand{\pevalname}{\calV}
\newcommand{\peval}[1]{\pevalname \, #1}
\newcommand{\invpevalname}{\pevalname^{\mbox{-} \mkern-2mu \mbox{\scriptsize 1}}}
\newcommand{\invpeval}[1]{\invpevalname \mkern-1.5mu #1}

\newcommand\slcssubname{\mathtt{slcs}}
\newcommand\slcseq{\simeq_{\slcssubname}}
\newcommand\SMCSm{\SLCS^-}
\newcommand\slcsmsubname{\slcssubname^-}
\newcommand\slcsmeq{\simeq_{\slcsmsubname}}
\newcommand\topochecker{{\tt topochecker}}
\newcommand\voxlogica{{\tt VoxLogicA}}

\newcommand{\INFLOG}{IML\xspace}
\newcommand{\linfand}[1]{\bigwedge_{#1}}
\newcommand\inflogname{\mathtt{IML}}
\newcommand\inflogeq{\simeq_{\inflogname}}

\newcommand{\closurefunctor}{\textbf{C}}

\begin{document}

\mainmatter



\title{Towards Spatial Bisimilarity for Closure Models: \\
Logical and Coalgebraic Characterisations\thanks{Research partially supported by the MIUR Project PRIN 2017FTXR7S ``IT-MaTTerS'' (Methods and Tools for Trustworthy Smart Systems).}}
\titlerunning{Spatial Bisimilarity}

\author{Vincenzo~Ciancia\inst{1} \and Diego~Latella\inst{1}  \and Mieke~Massink\inst{1} \and Erik~de~Vink\inst{2}}
\institute{
CNR-ISTI, Pisa, Italy,\\
\email{$\{$V.Ciancia, D.Latella, M.Massink$\}$@cnr.it}
\and
Eindhoven University of Technology, The Netherlands,\\
\email{evink@win.tue.nl}
}

\authorrunning{Ciancia et al.}

\maketitle

\pagestyle{plain}

\begin{abstract}
The topological interpretation of modal logics provides descriptive
languages and proof systems for reasoning about points of topological
spaces. Recent work has been devoted to model checking of 
\emph{spatial logics} on \emph{discrete} spatial structures, such as
finite graphs and digital images, with applications in various case
studies including medical image analysis. These recent developments
required a generalization step, from topological spaces to
\emph{closure spaces}. In this work we initiate the study of
bisimilarity and minimization algorithms that are consistent with the
closure spaces semantics. For this purpose we employ coalgebraic
models. We present a coalgebraic definition of bisimilarity for
quasi-discrete models, which is adequate with respect to
a spatial logic with reachability operators, complemented by a free
and open-source minimization tool for finite models. We also discuss the non-quasi-discrete
case, by providing a generalization of the well-known set-theoretical
notion of \emph{topo-bisimilarity}, and a categorical definition, in
the same spirit as the coalgebraic rendition of \emph{neighbourhood
  frames}, but employing the covariant power set functor, instead of
the contravariant one. We prove its adequacy with respect to
\emph{infinitary modal logic}.

\end{abstract}


\input{Introduction}
\input{Preliminaries}

\input{CMBisimilarity}
\input{CMCoalgebraically}
\input{SLCSEq}

\input{Tool}
\input{Generalisation}

\input{Discussion}

\bibliographystyle{splncs03}
\bibliography{BisLogCoa}

\vfill \mbox{}
\pagebreak

\appendix
\input{APX}

\end{document}

%% file: Introduction.tex
\section{Introduction}
\label{sec:Introduction}

Traditional modal logic enjoys a topological interpretation, according to which
the modal formula $\diamond\mkern1mu \Phi$ is true at a point~$x$ of a topological space,
whenever $x$ belongs to the \emph{topological closure} of the set of points at
which $\Phi$ is true. This fundamental observation has led to a variety of
extensions of the basic framework, with different proof systems and
computational properties, cf.~\cite{HBSL}.

\emph{Model checking} has been studied for the case of spatial logics only
recently. In order to retain the topological flavour, but aiming at analysis of
more general structures, also encompassing graphs, the \emph{Spatial Logic for
Closure Spaces} ($\SLCS$) has been proposed by Ciancia et al.\ in~\cite{CLLM14}
together with an algorithm for model checking of finite models. The logic
$\SLCS$ is interpreted on \emph{closure spaces} (a generalization of topological
spaces where the closure operator is not necessarily idempotent). We refer the
reader to~\cite{CLLM16} for a full account of the logic and its main features
and properties -- including its extension with a \emph{collective} fragment. The
logic and its model checkers \topochecker~\cite{CGLLM15} and
\voxlogica~\cite{BCLM19} have been applied to several case
studies~\cite{CLLM16,CGLLM15,CGGLLM18} including a declarative approach to
medical image analysis~\cite{BCLM16,BCLM19,BCLM19a,BBCLM19}. An encoding of the
discrete Region Connection Calculus RCC8D of~\cite{RLG13} into the collective
variant of $\SLCS$ has been proposed in~\cite{CLM19}. The logic has also
inspired other approaches to spatial reasoning in the context of signal temporal
logic and system monitoring~\cite{BBLN17,NBCLM18} and in the verification of cyber-physical systems \cite{TKG17}.

In this work, we initiate the study of bisimilarity and minimization algorithms
for spatial structures, employing equivalence relations on points of a closure
space that are \emph{adequate} with respect to spatial logical equivalence. That
is, we require that two points are bisimilar if and only if they satisfy the
same formulas of a (chosen) spatial logical language. For the topological case,
one such equivalence has been provided by Aiello and Van Benthem~\cite{BeB07},
under the name of \emph{topo-bisimilarity}. This relation is adequate with
respect to logical equivalence of basic infinitary modal logic, i.e.\ a boolean
logic with one modal operator, and infinitary conjunction/\-disjunction. In
contrast, besides basic modalities, the logic $\SLCS$ features operators that
make use of reachability via paths of bounded and unbounded length (for
instance, the \emph{surrounded} and \emph{touch} operators of~\cite{CLLM16}).
Although the study of such operators has not been developed in full detail in
the classical spatial logics literature, they have proved useful in case
studies. For instance, the ability to identify two areas, characterised by given
logical formulas, that additionally are in contact with each other, while
retaining the point-based approach of topo-logics, has been the key to derive a
segmentation algorithm that labels brain tumours in three-dimensional medical
images, with accuracy in par with manual segmentation, and best-in-class machine
learning methods~\cite{BCLM19}.

In the present paper, we focus on two different, related problems. First of all,
we identify a spatial definition of bisimilarity for quasi-discrete models
(those that correspond to graphs), and a minimization algorithm for
\emph{finite} models, in the setting of logics with reachability. This is
directly aimed at supporting the future developments of the spatial model
checking methodology that is currently in use, e.g.\ in~\cite{BCLM19}. In
Section~\ref{sec:CMBisimilarity} we present a set-theoretical definition, and
provide some examples. In Section~\ref{sec:CMCoalgebraically}, we provide a
coalgebraic rendition of such an equivalence. In Section~\ref{sec:SLCSEq}, we
prove adequacy with respect to logical equivalence of a logic with two
reachability operators (corresponding to the two directions of ``reaching'' and
``being reached''). In Section~\ref{sec:tool} we introduce an open source tool
that is able to minimize finite models via coalgebraic partition refinement.

The second research question that we address here, is whether the theory of
topo-bisimilarity of~\cite{BeB07}, characterising infinitary modal logic
(without reachability operators), can be generalised to closure models (not
limited to the quasi-discrete ones). In Section~\ref{sec:generalization}, we
first provide a consistent generalization, obtained by appropriately replacing
the notion of an \emph{open neighbourhood} with one that is equivalent in the
restricted setting of topological spaces, but not in the more general one. The
defined equivalence relation is adequate for infinitary modal logic when
interpreted on closure spaces. Then, we provide a coalgebraic definition. We
prove that logical equivalence of infinitary modal logic can be characterised as
behavioural equivalence for coalgebras of the \emph{closure functor}
$\pws(\pws(-))$. The notion we propose is similar in spirit to
\emph{neighbourhood frames} (see~\cite{HKP09}), although we use the
\emph{covariant} power set, therefore staying closer to the more classical
 literature on coalgebras in Computer Science.
 
Although the results we present are sound and stable, we consider them as a
preliminary foundation. Future work will be devoted to the characterisation of
logical equivalence for variants of the considered logics (for instance, those
that cannot express one-step modalities, logics with distances, etc.). We
provide some discussion on these matters in Section~\ref{sec:discussion}.

%% file: Preliminaries.tex
\section{Preliminaries}
\label{sec:Preliminaries}

Given set~$X$ and relation $R \subseteq X \times X$, we let $R^{\,t}$
denote the \emph{transitive closure} of~$R$ and let $R^{-1}$ denote
the inverse of~$R$, i.e.\ $R^{-1} = \ZET{(x_1,x_2)}{(x_2,x_1)\in R}$.
For $x \in X$, we let $[x]_R$ denote the equivalence class of~$x$ (we
will omit the subscript whenever this does not cause confusion).  We
let $\pws{}$ denote the covariant powerset functor; for $f : X \to Y$
and $A \subseteq X$, its action on arrows $\pws{}f \, A$, often
abbreviated to $f \, A$, is defined as $\ZET{fa}{a \in A}$.
Similarly, $\fpws{X}$ denotes the covariant \emph{finite} powerset
functor. For $f : X \to Y$ a function, we denote by $f^{-1} : \pws Y \to \pws X$ its ``relational'' inverse, that is the function mapping $B \subseteq Y$ to $\{ x \in X \mid f x \in B \}$. We will often use currying for function type definitions and
applications, when this does not create confusion.

\begin{definition}\label{def:ClosureSpaces}
 A \emph{closure space} is a pair $(X,\closurename{})$ where $X$ is a
 non-empty set (of {\em points}) and \emph{$\closurename{} : \pws{X} \to
   \pws{X}$} is a function satisfying the following axioms.
 \begin{enumerate}
 \item $\closure{}{\emptyset} = \emptyset$
 \item $A \subseteq \closure{}{A}$ for all  $A \subseteq X$
 \item $\closure{}{(A_1 \cup A_2)} = \closure{}{A_1} \cup
   \closure{}{A_2}$ for all $A_1,A_2\subseteq X$
 \end{enumerate}
\end{definition}

\noindent
The definition of a closure space goes back to Eduard Čech.  By the Kuratowski
definition, topological spaces coincide with the sub-class of closure spaces for
which also the \emph{idempotence} axiom $\closure{}{(\closure{}{A})} =
\closure{}{A}$ holds. The \emph{interior} operator is the dual of closure:
$\interior{}{A} = \overline{\closure{}{(\overline{A})}}$.
%
%
Given a relation $R \subseteq X \times X$, the function
$\closurename{R}: \pws{X} \to \pws{X}$ with \emph{$\closure{R}(A) = A
  \cup \ZET{x}{\exists \mkern1mu a \in A \colon a \, R \,x}$}
satisfies the axioms of Definition~\ref{def:ClosureSpaces}, thus
making $(X,\closurename{R})$ a closure space. We say that
$(X,\closurename{R})$ is \emph{based on}~$R$. It can be shown that the
sub-class of closure spaces that can be generated by a relation as
above coincides with the class of \emph{quasi-discrete} closure
spaces, i.e.\ closure spaces where every $x \in X$ has a minimal
neighbourhood or, equivalently, for each $A \subseteq X, \,
\closure{}{A} = \bigcup_{a \in A} \closure{}{\SET{a}}$. Thus discrete
structures, like graphs or Kripke structures can be seen
as quasi-discrete closure spaces. With reference to a quasi-discrete
closure space $(X,\closurename{R})$ based on a relation~$R$, we define
the abbreviations $\forwclosurename, \backclosurename : X \to \pws{X}$
by $\forwclosure{x} = \closure{R}(\SET{x})$ and $\backclosure{x} =
\closure{{R^{-1}}}(\SET{x})$.

\begin{definition}\label{def:ContinuousFunction}
  A \emph{continuous function} from closure space
  $(X_1,\closurename{1})$ to closure space $(X_2,\closurename{2})$ is
  a function $f : X_1 \to X_2$ such that, for all sets $A \subseteq
  X_1$, it holds that $f (\closure{1}{A}) \subseteq \closure{2}{(f\,A)}$.
\end{definition}

\noindent
We fix a set $AP$ of \emph{atomic predicates}.  A closure model
$\model=((X,\closurename{}), \pevalname)$ is a pair with
$(X,\closurename{})$ a closure space, and $\pevalname: AP \to \pws{X}$
the (atomic predicate) valuation (function). We define $\invpevalname:
\pws{X} \to \pws{AP}$ with \emph{$\invpeval{A} = \ZET{p \in
    AP}{\exists \mkern1mu a \in A \colon a \in \peval p}$} and we let
$\invpeval{}{x}$ abbreviate $\invpeval{}{\SET{x}}$. We say that a
closure model $\model = ((X,\closurename{}), \pevalname)$ is
quasi-discrete if $(X,\closurename{})$ is quasi-discrete.  A
quasi-discrete closure model $((X,\closurename{R}), \pevalname)$ is
\emph{finitely closed} if $\forwclosure{}{x}$ is finite for all $x \in
X$. Similarly, we say that $\model$ is \emph{finitely backward closed}
if $\backclosure{}{x}$ is finite for all $x \in X$.

In the following definition, $(\nats, \closurename{\succ})$ is the
quasi-discrete closure space of the natural numbers $\nats$ with the
{\em successor} relation $\succ$.

\begin{definition}
  \label{def:qspath}
  A quasi-discrete \emph{path} $\pi$ in $(X,\closurename{})$ is a
  continuous function from $(\nats, \closurename{\succ})$ to
  $(X,\closurename{})$.
\end{definition}


\noindent
We recall some basic definitions from coalgebra. See e.g.~\cite{Rut00}
for more details. For a functor $\calF : \Set \to \Set$ on the
category $\Set$ of sets and functions, a \emph{coalgebra}~$\calX$
of~$\calF$ is a set~$X$ together with a mapping $\alpha : X \to \calF
X$.  A \emph{homomorphism} between two $\calF$-coalgebras $\calX =
(X,\alpha)$ and~$\calY = (Y,\beta)$ is a function $f : X \to Y$ such
that $(\calF f) \compose \alpha = \beta \compose f$. An
$\calF$-coalgebra $(\Omega_{\calF},\omega_{\calF})$ is called
\emph{final}, if there exists, for every $\calF$-coalgebra
$\calX=(X,\alpha)$, a unique homomorphism
$\fmorph{\calF}{\calX}{\cdot} : (X,\alpha) \to
(\Omega_{\calF},\omega_{\calF})$. Two elements~$x_1, x_2$ of an
$\calF$-coalgebra~$\calX$ are called {\em behavioural equivalent} with
respect to~$\calF$ if $\fmorph{\calF}{\calX}{x_1} =
\fmorph{\calF}{\calX}{x_2}$, denoted $x_1 \beq{\calF}{\calX} x_2$.  In
the notation $\fmorph{\calF}{\calX}{\cdot}$ as well as
$\beq{\calF}{\calX}$, the indication of the specific coalgebra~$\calX$
will be omitted when clear from the context.
A functor~$\calF$ is called \emph{$\kappa$-accessible} if it preserves
$\kappa$-filtered colimits for some cardinal number~$\kappa$. However,
in the category~$\Set$, we have the following characterization of
accessibility: for every set~$X$ and any element $\xi \in \calF X$,
there exists a subset $Y \subseteq X$ with ${|} Y {|} < \kappa$, such
that $\xi \in \calF Y$. It holds that a functor has a final coalgebra
if it is $\kappa$-accessible for some cardinal
number~$\kappa$. See~\cite{AP04:tcs}.

%% file: CMBisimilarity.tex
\section{Bisimilarity for Quasi-discrete Closure Models}
\label{sec:CMBisimilarity}

In this section we give a back-and-forth definition of bisimilarity in
quasi-discrete closure spaces, and an alternative
characterization that makes explicit use of the underlying closure.

\begin{definition}
  \label{def:CMB} 
  Given quasi-discrete closure model $\model=((X,\closurename{R}),
  \pevalname)$ based on~$R$, a non-empty relation $B \subseteq X
  \times X$ is a \emph{bisimulation relation} if for all $x_1,x_2 \in
  X$ such that $(x_1,x_2) \in B$, all five conditions below hold:
  \begin{enumerate}
  \item $\invpeval{x_1}=\invpeval{x_2}$
  \item for all $x_1' \in \forwclosure{x_1}$ there exists $x_2' \in
    \forwclosure{x_2}$ such that $(x_1',x_2') \in B$
  \item for all $x_2' \in \forwclosure{x_2}$ there exists $x_1' \in
    \forwclosure{x_1}$ such that $(x_1',x_2') \in B$
  \item for all $x_1' \in \backclosure{x_1}$ there exists $x_2' \in
    \backclosure{}{x_2}$ such that $(x_1',x_2') \in B$ 
  \item for all $x_2' \in \backclosure{}{x_2}$ there exists $x_1' \in
    \backclosure{}{x_1}$ such that $(x_1',x_2') \in B.$
  \end{enumerate}
  We say that $x_1$ and~$x_2$ are \emph{bisimilar} (written $x_1
  \cmbisim^{\model} x_2$) if there exists a bisimulation relation~$B$
  for~$X$ such that $(x_1,x_2) \in B$.
\end{definition}

\noindent
In the sequel, for the sake of notational simplicity, we will write
$\cmbisim$ instead of $\cmbisim^{\model}$ whenever this does not cause
confusion.

  \begin{figure}[h]
    \centering
    \centerline{\resizebox{0.7in}{!}
      {\includegraphics{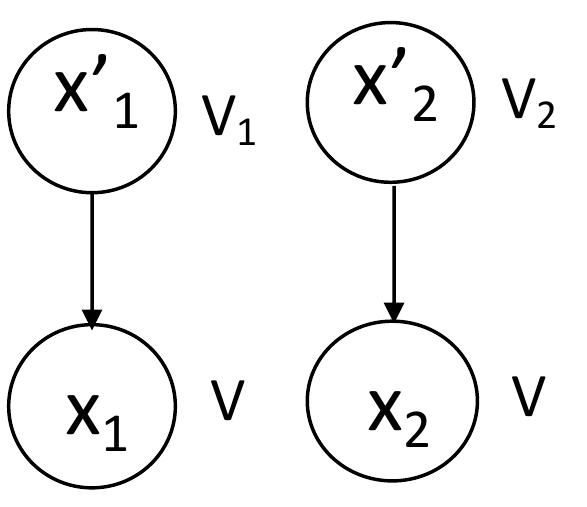}}}
    \caption{A model}\label{fig:RemOneCounterex}
  \end{figure}

\begin{remark}
  \label{rem:BaFbis}
  Bisimilarity for quasi-discrete closure models is reminiscent to strong
  back-and-forth bisimilarity~\cite{De+90} and is stronger than a
  spatial version of standard bisimilarity that would include only
  items 1 to~3 above.

  In order to illustrate this, consider the model~$\model$ of
  Figure~\ref{fig:RemOneCounterex}, where $\model =
  ((X,\closurename{}),\peval)$ with $X = \SET{x_1,x_2,x'_1, x'_2}$,
  the closure operator~$\closurename{}$ defined by
  $\closure{}{x_j}=\SET{x_j}$ and $\closure{}{x'_j}=\SET{x_j,x'_j }$
  for $j = 1,2$ and valuation~$\peval$ such that $\invpeval x'_1 \neq
  \invpeval x'_2$ but $\invpeval x_1 = \invpeval x_2$. Furthermore,
  let $B$ be the reflexive and symmetric closure of
  $\SET{(x_1,x_2)}$. Then $B$ would be a standard bisimulation showing
  the points $x_1$ and~$x_2$ bisimilar, since only items 1 to 3 of
  Definition~\ref{def:CMB} are considered.  However, we have $x_1
  \not\cmbisim x_2$ according to Definition~\ref{def:CMB}, because of
  items 4 and~5. As we will see in Section~\ref{sec:SLCSEq}, this is
  directly related to the semantics of logic operator~$\lrcdname$.
\end{remark}

\noindent
Given a quasi-discrete closure model $\model=((X,\closurename{R}),
\pevalname)$, it is easy to see that $\cmbisim^{\model}$ is an
equivalence relation and it is itself a bisimulation relation, namely
the union of all bisimulation relations, i.e.\ the largest (coarsest)
bisimulation relation.

In the following, we provide an alternative, equivalent, definition of
bisimilarity, which will prove useful for the developments in
Section~\ref{sec:CMCoalgebraically} and Section~\ref{sec:SLCSEq}.


\begin{definition}
  \label{def:altCMB}
  Given quasi-discrete closure model $\model=((X,\closurename{R}),
  \pevalname)$ based on $R$, a non-empty equivalence relation $B
  \subseteq X \times X$ is a \emph{bisimulation relation} if for all
  $x_1,x_2 \in X$ such that $(x_1,x_2) \in B$ it holds that
  \begin{enumerate}
  \item $\invpeval{x_1}=\invpeval{x_2}$, and
    \smallskip
  \item for all equivalence classes $C \in X/B$ both the following conditions hold:
    \begin{enumerate}
    \item $(\forwclosure{x_1}) \cap C \not= \emptyset$ iff
      $(\forwclosure{x_2}) \cap C \not= \emptyset$
    \item $(\backclosure{x_1}) \cap C \not= \emptyset$ iff
      $(\backclosure{}{x_2}) \cap C \not= \emptyset.$
    \end{enumerate}
  \end{enumerate}
  We say that $x_1$ and~$x_2$ are \emph{bisimilar}, notation $x_1
  \altcmbisim^{\model} x_2$) if there exists a bisimulation
  relation~$B$ such that $(x_1,x_2) \in B$.
\end{definition}

\noindent
In the following, for the sake of notational simplicity, we will write
$\altcmbisim$ instead of $\altcmbisim^{\model}$ whenever this does not
cause confusion.

Given a quasi-discrete closure model $\model = ((X,\closurename{R}),
\pevalname)$, also for~$\altcmbisim^{\model}$ it is easy to see that
it is an equivalence relation and that it is in fact the largest
(coarsest) bisimulation relation. In addition, it is straightforward
to show that $\cmbisim^{\model}$ is a bisimulation relation according
to Definition~\ref{def:altCMB}. So, ${\cmbisim^{\model}} \, \subseteq
\, {\altcmbisim^{\model}}$. Moreover, it also holds that
$\altcmbisim^{\model}$ is a bisimulation relation according to
Definition~\ref{def:CMB} and therefore ${\altcmbisim^{\model}} \,
\subseteq \, {\cmbisim^{\model}}$. Consequently the two equivalences
coincide.

\begin{figure}
\centerline{\resizebox{4cm}{!}{\includegraphics{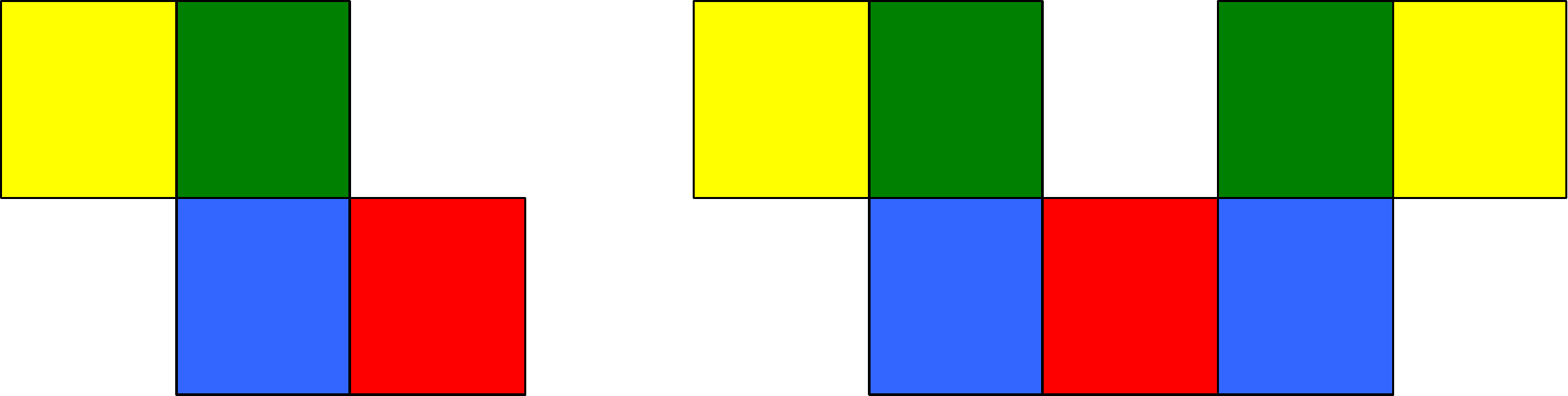}}}
\caption{Model $\model_1$.}\label{fig:ExaUno}
\end{figure}

\begin{example}
  \label{exa:ExaUno}
  In Figure~\ref{fig:ExaUno} a quasi-discrete closure model $\model_1
  = ((X_1,\closurename{R_1}),\peval_{\!1})$ is shown where $X_1$
  contains 11 elements, each represented by a coloured square box,
  which we call a \emph{cell}. The relation~$R_1$ is the so-called
  \emph{orthogonal adjacency} relation~\cite{RLG13}, i.e.\ the
  reflexive and symmetric relation such that two cells are related iff
  they share an edge. Note, $X_1$ is not path-connected.
  The set $AP$ of atomic predicates is the set $\SET{\mathtt{red},
    \mathtt{blue}, \mathtt{green}, \mathtt{yellow}}$ and
  $\peval_{\!1}$ associates each predicate (i.e.\ colour) to the set
  of cells of that specific colour\footnote{Spaces like $\model_1$ can
    be thought of as digital images where each cell represents a
    distinct pixel and the background of the image has been filtered
    out.}; in this example, each cell satisfies
  exactly one atomic proposition. The two red cells are
  $\cmbisim^{\model_1\!\!}-$bisimilar. In order to see this, 
  consider the relation~$B_1$ which is the minimal reflexive and
  symmetric binary relation on~$X_1$ such that
  \begin{itemize}
  \item the two red points are related;
  \item the blue (green, yellow, respectively) point of the left-hand
    component is related to each blue (green, yellow, respectively)
    point of the right-hand component.
  \end{itemize}
It is easy to see that $B_1$ satisfies the conditions of
Definition~\ref{def:CMB}. For instance, the (forward and backward) closure of
the left-hand side red cell contains only the cell itself and the blue adjacent
one and, for each such cell, there is one in the right-hand side of the same
colour and related to the former by~$B_1$. Similarly, the closure of the
right-hand side red cell contains only the cell itself and the two blue adjacent
ones and, for each such cell, there is one in the left-hand side with the same
colour and related to the former by~$B_1$. Similar reasoning applies to all
other pairs of~$B_1$. Finally, the two red cell are related by bisimulation
relation~ $B_1$ and so they are bisimilar. 
\end{example}

\begin{figure}
\centering
\subfloat[][]
{
\includegraphics[height=0.7cm]{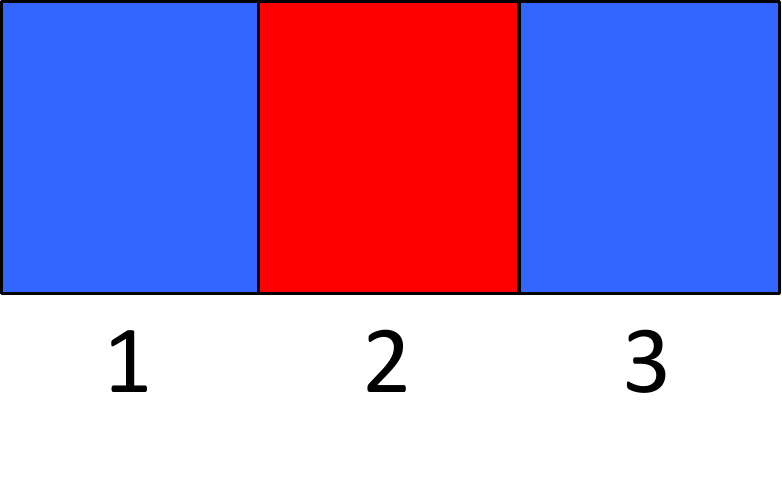}
\label{fig:MatOneTree}
}\hspace{0.1in}
\centering
\subfloat[][]
{
\includegraphics[height=1.6cm]{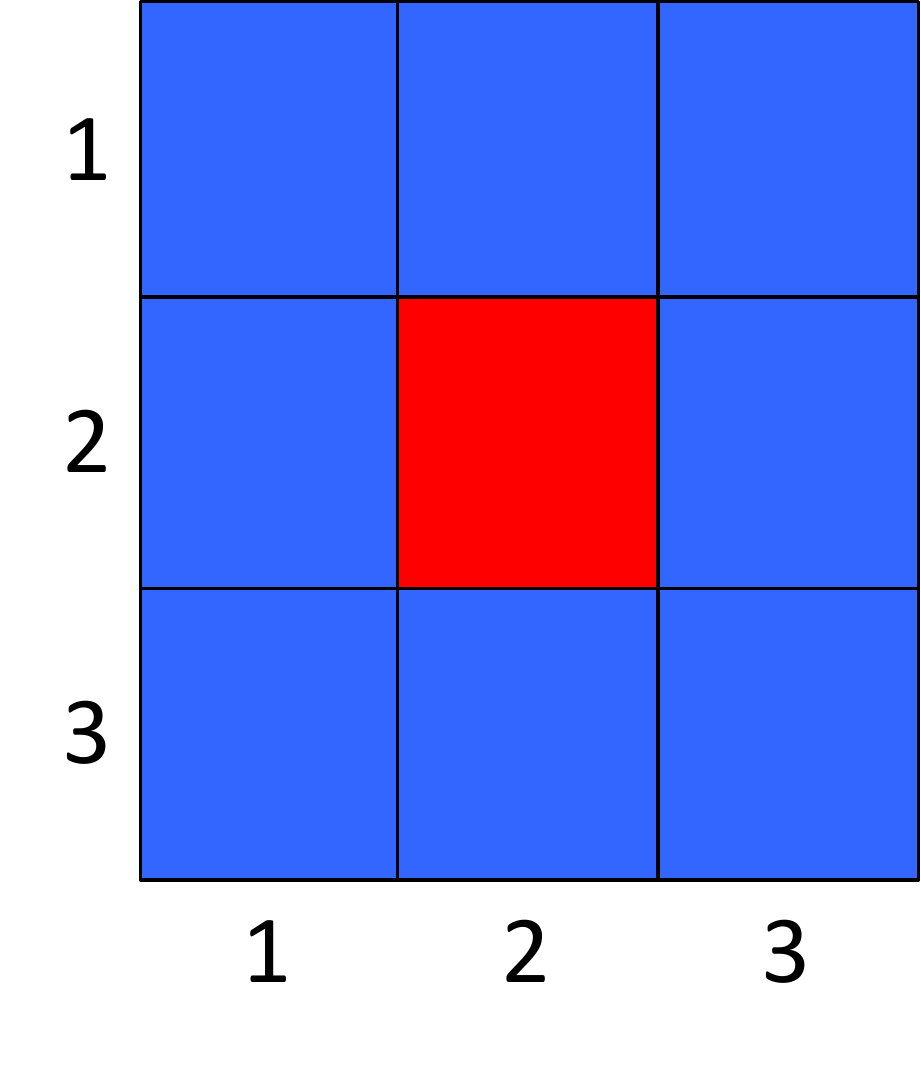}
\label{fig:MatThreeThree}
}\hspace{0.1in}
\centering
\subfloat[][]
{
\includegraphics[height=2cm]{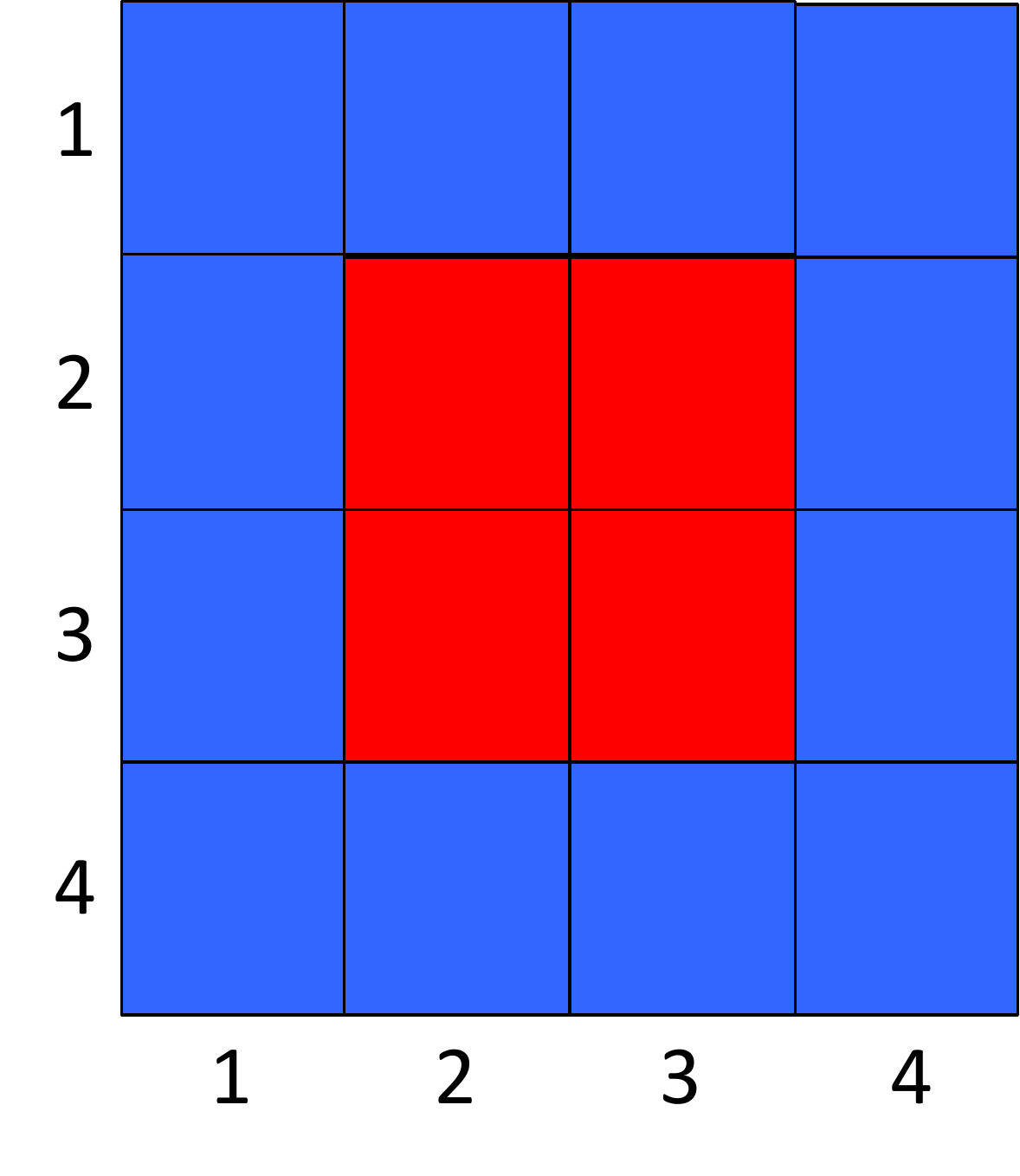}
\label{fig:MatFourFour}
}\hspace{0.1in}
\subfloat[][]
{
\includegraphics[height=3cm]{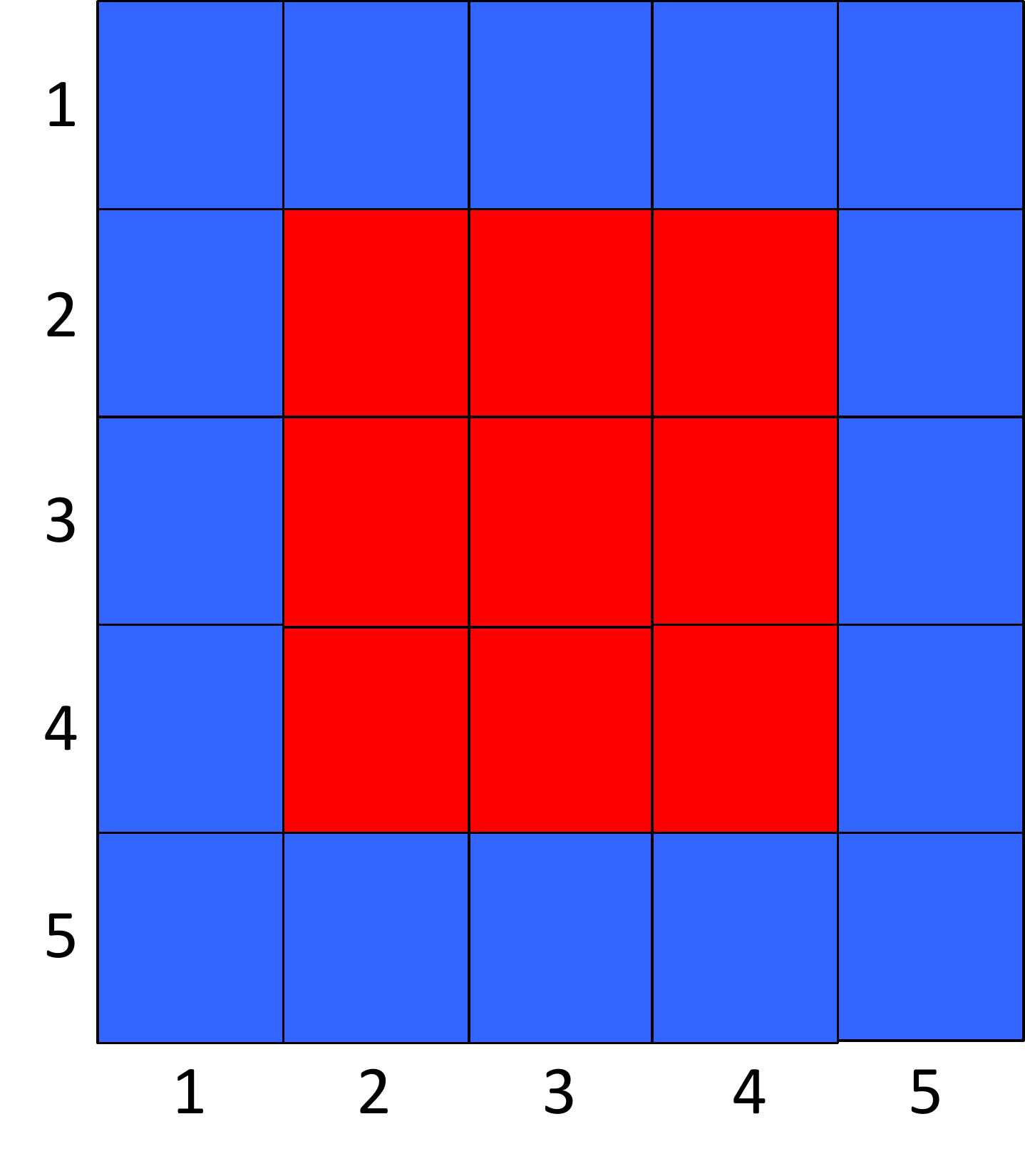}
\label{fig:MatFiveFive}
}
\caption{Models $\model_a$ (a), $\model_b$ (b),
$\model_c$ (c), and $\model_d$ (d).}\label{fig:ExaDue}
\end{figure}

\begin{example}
  \label{exa:ExaDue}
  For the sake of exposition, in this example, we will use matrix
  notation for referring to cells in Figure~\ref{fig:ExaDue}. Thus,
  $a_{2}$ and~$b_{22}$ refer to the red cells in
  Figure~\ref{fig:ExaDue}(a) and~\ref{fig:ExaDue}(b), respectively,
  whereas the four 
  red cells in Figure~\ref{fig:ExaDue}(c) are referred to as $c_{22},
  c_{23}, c_{32}, c_{33}$ and so on. We consider four models $\model_j
  = ((X_j,\closurename{R_j}),\peval_{\!j})$, for $j = a,b,c,d$.
  with
%
$X_a = \ZET{a_i}{1 \leqslant i \leqslant 3}$,
$X_b = \ZET{b_{ij}}{1 \leqslant i,j \leqslant 3}$,
$X_c = \ZET{c_{ij}}{1 \leqslant i,j \leqslant 4}$,
$X_d = \ZET{d_{ij}}{1 \leqslant i,j \leqslant 5}$.
Differently from Example~\ref{exa:ExaUno}, in each of the four models
$\model_j$, we assume an \emph{orthodiagonal adjacency
  relation}~$R_j$, namely, the reflexive and symmetric relation such
that two cells in $\model_j$ are related iff they share an edge \emph{or} a vertex~\cite{RLG13}. For instance, $\SET{(c_{11},c_{12}),
  (c_{11},c_{22})} \subseteq R_c$. This choice of the adjacency
relation simplifies the description of the example\footnote{In fact,
  as we will see, using the orthodiagonal relation makes the corner
  cells of $\model_j$, for $j\in \SET{b,c,d}$, bisimilar to all other
  blue cells of the model, which would not be the case, had we used the orthogonal relation.}.  
The set~$AP$ of atomic propositions is the set $\SET{\mathtt{red},
  \mathtt{blue}}$, with
$\peval_{\!a} \,\mathtt{red} = \SET{a_2}$,
$\peval_{\!b} \,\mathtt{red} = \SET{b_{22}}$,
$\peval_{\!c} \,\mathtt{red} = \bigcup_{i,j=2,2}^{3,3}\SET{c_{ij}}$, and
$\peval_{\!d} \,\mathtt{red} = \bigcup_{i,j=2,2}^{4,4}\SET{d_{ij}}$, with
$\peval_{\!j} \,\mathtt{blue} = X_j \setminus (\peval_{\!j} \,\mathtt{red})$ for $j \in \SET{a,b,c,d}$.
We use the shorthand 
$\gamma \mkern1mu x = \pevalname^{\mbox{\,\scriptsize --1}}_j
\mkern1mu x$ for $x \in X_j$ and $j = a,b,c,d$.
Note, since $R_j$ is reflexive, $\gamma \mkern1mu x$ includes the
point~$x$ itself.
%
%
Moreover, the relations~$B_j$ are defined as $B_j = \ZET{(x,y) \in
  X_j^2}{\gamma \mkern1mu x = \gamma \mkern1mu y}$, for $j \in
\SET{a,b,c}$.

Clearly $a_1 \altcmbisim^{\model_a} a_3$ since $(a_1,a_3) \in B_a$,
which is a bisimulation. Take for instance $\closure{R_a}{\!a_1} =
\SET{a_1,a_2}$ and note that $(a_1,a_3) \in B_a$ with $a_3 \in
\closure{R_a}{\!a_3}$ and, similarly, $(a_2,a_2) \in B_a$ with $a_2
\in \closure{R_a}{\!a_3}$. The reasoning for the other cells of~$X_a$
is similar.

It is also easy to see that for all $x,y \in X_b$ we have $x
\altcmbisim^{\model_b} y$ iff $\gamma \mkern1mu x = \gamma \mkern1mu
y$; for instance, for each element~$z_1$ of~$\closure{R_b}{b_{11}} =
\SET{b_{11}, b_{12}, b_{21},b_{22}}$ there exists an element~$z_2$
of~$\closure{R_b}{b_{23}} = \SET{b_{12},b_{13}, b_{22},
  b_{23},b_{32},b_{33}}$ such that $(z_1,z_2) \in B_b$, which is a
bisimulation containing also $(b_{11}, b_{23})$.

Let us now consider the model $\model_{ab}=((X_a \cup X_b,
\closurename{(R_a \cup R_b)}),\peval_{\!ab})$ where $\peval_{\!ab} \,p
= {\peval_{\!a} \mkern1mu p} \mkern1mu \cup \mkern1mu {\peval_{\!b}
  \mkern1mu p}$. We define the relation $B_{ab} \subseteq (X_a \cup
X_b)$ by $B_{ab} =\ZET{(x,y) \in X_a \cup X_b}{\gamma \mkern1mu x =
  \gamma \mkern1mu y}$. The reader is invited to prove that $a_2
\altcmbisim^{\model_{ab}} b_{22}$.
Similar reasoning shows that $B_c$ is a bisimulation for~$\model_c$
and that the quotient $X_{c/\!\!\altcmbisim^{\model_c}} = \SET{\ZET{x \in
    X_c}{\gamma \mkern1mu x = \mathtt{red}},\ZET{x \in X_c}{\gamma
    \mkern1mu x = \mathtt{blue}}}$ is a two-element set.
Also, for model $\model_{bc} = (((X_b \cup X_c), \closurename{R_b \cup
  R_c}), \peval_{\!bc})$, with $\peval_{\!bc} \mkern1mu p =
\peval_{\!b} \mkern1mu p \cup \peval_{\!c} \mkern1mu p$, we have
$c_{ij} \altcmbisim^{\model_{bc}} b_{22}$ for $i,j \in \SET{2,3}$, due
to the existence of the bisimulation $B_{bc} = \ZET{(x,y) \in X_b \cup
  X_c}{\gamma \mkern1mu x = \gamma \mkern1mu y}$.

In general, if we take as a model the union~$\model_{abc}$ of $\model_a$,
$\model_b$ and~$\model_c$, i.e.\ $\model_{abc}= (((X_a \cup X_b \cup X_c)),
\closurename{R_a \cup R_b \cup R_c}, \peval_{\!abc})$, with $\peval_{\!abc}
\mkern1mu p = \peval_{\!a} \mkern1mu p \cup \peval_{\!b} \mkern1mu p \cup
\peval_{\!c} \, p$, we can easily see that all blue cells are bisimilar to one
another and all red cells are bisimilar to one another (and no blue cell is
bisimilar to any red one). In fact, as hinted above, there exists a minimal
model with just two cells, one red and one blue, that are adjacent (in this
case, the orthogonal and the orthodiagonal relations coincide); all red (blue,
respectively) cells of $\model_{abc}$ are equivalent to the red (blue,
respectively) cell of the minimal model.

Finally, let us consider cell~$d_{33}$ and, say, $d_{22}$ in
model~$\model_d$.  It is easy to see that $d_{22}
\not\altcmbisim^{\model_d}d_{33}$. In fact, $\closure{R_d}{d_{33}} =
\ZET{x\in X_d}{\gamma \mkern1mu x = \mathtt{red}}$ and there is $y \in
\closure{R_d}{d_{22}}$ such that $\gamma \,y = \mathtt{blue}$. Thus, any
bisimulation~$B$ should include $(y,z)$ with $z \in
\closure{R_d}{d_{33}}$, because of the transfer condition 2 of
Definition~\ref{def:CMB} with respect to $(d_{22},d_{33})$, but this
is impossible because $(y,z) \in B$ would violate condition 3 of
Definition~\ref{def:CMB} because $\gamma \mkern1mu y =
\mathtt{blue} \neq \mathtt{red} = \gamma \mkern1mu z$.
\end{example}

%
%
%

%% file: CMCoalgebraically.tex
\section{Quasi-discrete Closure Models Coalgebraically}
\label{sec:CMCoalgebraically}


Let the quasi-discrete closure model $\model=((X,\closurename{}), \peval)$ be
finitely closed and finitely backward closed (but not necessarily finite).  We
can represent~$\model$ as a pair $(X,\eta_{\model})$, with the function
$\eta_{\model} : X \to (\fpws{AP}) \times (\fpws{X}) \times (\fpws{X})$ such
that $\eta_{\model} x = (\invpeval{x}, \forwclosure{x}, \backclosure{x})$. In
the sequel we will write~$\eta$ instead of~$\eta_{\model}$, when this does not
cause confusion. Note that not all pairs~$(X,\eta)$ represent closure models,
but only those for which $(\eta \,x)_3 = \ZET{x' \in X}{x \in (\eta x')_2}$ and
$(\eta \,x)_2 =\ZET{x' \in X}{x \in (\eta x')_3}$.

\begin{example}
  Model~$\model_a$ of Figure~\ref{fig:ExaDue} corresponds to the pair
  $(\SET{a_1,a_2,a_3},\eta_a)$ with 
  $\eta_a \, a_1 =(\SET{\mathtt{blue}}, \SET{a_1,a_2}, \SET{a_1,a_2})$,
  $\eta_a \, a_2 =(\SET{\mathtt{red}}, \SET{a_1,a_2,a_3},
  \SET{a_1,a_2,a_3})$, and 
  $\eta_a \, a_3 =(\SET{\mathtt{blue}}, \SET{a_2,a_3},
  \SET{a_2,a_3})$. The other models of the figure can be represented
  similarly. 
\end{example}

\noindent 
The following is a reformulation of bisimilarity in terms of the
function~$\eta_{\model}$. 

\begin{definition}\label{def:ETB} 
  An equivalence relation $B \subseteq X \times X$ is an
  $\eta$-bisimulation if $(x_1,x_2) \in B$ implies that the following
  holds:
  \begin{enumerate}
  \item $(\eta \, x_1)_1 = (\eta \, x_2)_1$, and
  \item for all $C \in X/B$ it holds that
    \begin{enumerate}
    \item $(\eta \, x_1)_2 \cap C \not=\emptyset \mbox{ iff } (\eta \,
      x_2)_2 \cap C \neq \emptyset$, and
    \item $(\eta \, x_1)_3 \cap C \neq \emptyset \mbox{ iff } (\eta \,
      x_2)_3 \cap C \neq \emptyset.$
    \end{enumerate}
  \end{enumerate}
  We say that $x_1$ and~$x_2$ are $\eta$-\emph{bisimilar}, notation $x_1
  \ebisim x_2$, if there exists an $\eta$-bisimulation relation~$B$ such
  that $(x_1,x_2) \in B$.
\end{definition}

\begin{example}
  With reference to model~$\model_a$ of Figure~\ref{fig:ExaDue} and
  $(\SET{a_1,a_2,a_3},\eta_a)$ above, it is easy to see that $a_1$
  and~$a_3$ are $\eta_a$-{bisimilar}.
\end{example}

\noindent
The following lemma follows directly from Definition~\ref{def:altCMB}
and Definition~\ref{def:ETB}. 

\begin{lemma}
  \label{lemma:ebisISaltcmbisim}
  $\ebisim$ coincides with $\altcmbisim$.
\end{lemma}

\begin{definition}
  \label{def:simplerfunctor}
  The functor $\calT: \Set \to \Set$ assigns to a set~$X$ the product
  set $\fpws{AP} \times \fpws{X} \times \fpws{X}$ and to a mapping $f
  : X \to Y$ the mapping $\calT \mkern1mu f : (\fpws{AP} \to \fpws{X}
  \to \fpws{X}) \to (\fpws{AP} \times \fpws{Y} \times \fpws{Y})$
  where, for all $v \in \fpws{AP}$ and $z,z' \in \fpws{X}$, $(\calT f)
  \mkern1mu v \mkern1mu z \mkern1mu z' = (v,(f z),(f z'))$.
\end{definition}

\noindent
Clearly, the model~$\model$, represented as $(X,\eta)$, can be interpreted
as a coalgebra  of functor~$\calT$.

\begin{lemma}
  \label{lemma:calTfinalco}
  The functor $\calT$ has a final coalgebra.
\end{lemma}

\begin{proof}
  Constants, finite products, and the finite powerset are
  $\omega$-accessible functors. The class of $\kappa$-accessible
  functors for any~$\kappa$ is closed with respect to composition, and
  $\kappa$-accessible functors have final coalgebras. \pfend
\end{proof}

\noindent
We recall that two elements~$x_1, x_2$ of an $\calT$-coalgebra~$\calX$
are behavioural equivalent if $\fmorph{\calT}{\calX}{x_1} =
\fmorph{\calT}{\calX}{x_2}$, denoted $x_1 \beq{\calT}{\calX} x_2$,
where $\fmorph{\calT}{\calX}{\cdot}$ is the unique morphism
from~$\calX$ to the final coalgebra of functor~$\calT$.

The following theorem shows that behavioural equivalence and $\eta$-bisimilarity coincide. The proof follows the same pattern as that of Theorem 4.3 in \cite{LMV15a}.

\begin{theorem}
  \label{theo:etabisbeh}
  \mbox{%
  Behavioural equivalence and $\eta$-bisimilarity coincide,
  i.e.\ $\beq{\calT}{} \, = \, \ebisim$.}
\end{theorem}

\begin{proof}
  Let $x_1, x_2 \in X$.  We
  first prove that ${x_1 \ebisim x_2}$ implies $x_1 \beq{\calT}{}
  x_2$.  So, assume ${x_1 \ebisim x_2}$.  Let $B \subseteq X \times X$
  be an $\eta$-bisimulation with $(x_1,x_2)\in B$ and recall that
  $(X,\eta)$ is a $\calT$-coalgebra. We turn the collection of
  equivalence classes~$X/B$ into a $\calT$-coalgebra
  $\model/B = (X/B,\varrho_B)$ where, for $s \in X$
  \begin{displaymath}
    \varrho_B \, [s]_B  = (
    (\eta \mkern1mu s)_1,
    \ZET{C\in X/B}{(\eta \mkern1mu s)_2 \cap C \neq \emptyset},
    \ZET{C\in X/B}{(\eta \mkern1mu s)_3 \cap C \neq \emptyset}) 
  \end{displaymath}
  This is well-defined since $B$ is an $\eta$-bisimulation: if
  $(s,s')\in B$ then 
  we have
  \begin{displaymath}
    \def\arraystretch{1.3}
    \begin{array}{rcll}
      \varrho_B \mkern1mu [s]_B 
      & = &
      ((\eta \mkern1mu s)_1, 
      \ZET{C\in X/B}{(\eta \mkern1mu s)_2 \cap C \neq \emptyset},
      \ZET{C\in X/B}{(\eta \mkern1mu s)_3 \cap C \neq \emptyset})
      \\ & & \qquad \text{(def.\ of $\varrho_B$)} 
      \\
      & = &
      ((\eta \mkern1mu s')_1,
      \ZET{C\in X/B}{(\eta \mkern1mu s')_2 \cap C \neq \emptyset},
      \ZET{C\in X/B}{(\eta \mkern1mu  s')_3 \cap C \neq \emptyset})
      \\ & & \qquad \text{($(s,s')\in B$; $B$ is an $\eta$-bisimulation)} 
      \\
      & = & \varrho_B \, [s']_B
      \qquad \text{(def.\ of $\varrho_B$).}
    \end{array}
    \def\arraystretch{1.0}
  \end{displaymath}
  The canonical mapping $\varepsilon_B : X \to X/B$ is a
  $\calT$-homomorphism, i.e.\ $(\calT \,\varepsilon_B) \compose \eta =
  \varrho_B \compose \varepsilon_B$ as can be verified as follows.
For  $s \in X$, we have
\begin{displaymath}
    \def\arraystretch{1.3}
    \begin{array}{rcll}
      (\calT\, \varepsilon_B)(\eta \mkern1mu s) 
      & = & ((
      \eta \mkern1mu s)_1, 
      \varepsilon_B \mkern1mu (\eta \mkern1mu s)_2, 
      \varepsilon_B \mkern1mu (\eta \mkern1mu s)_3 )
      & \text{(def.\ of $\calT$)} 
      \\ & = &
      ((\eta \mkern1mu s)_1, 
      \ZET{[t]_B}{t \in (\eta \mkern1mu s)_2},
      \ZET{[t]_B}{t \in (\eta \mkern1mu s)_3})
      \qquad \mbox{} & 
      \text{(def.\ of $\varepsilon_B$)}
      \\ & = &
      \multicolumn{2}{l}{%
        ((\eta \mkern1mu s)_1, 
        \ZET{C\in X/B }{(\eta \mkern1mu s)_2 \cap C \neq \emptyset}, 
        \ZET{C\in X/B }{(\eta \mkern1mu s)_3 \cap C \neq \emptyset})
      } 
      \\
      & = & \varrho_B [s]_B & \text{(def.\ of $\varrho_B$)}
      \\
      & = & \varrho_B (\varepsilon_B\, s) & \text{(def.\ of $\varepsilon_B$).}
    \end{array}
    \def\arraystretch{1.0}
  \end{displaymath}
  Thus, $(\calT \,\varepsilon_B) \compose \eta = \varrho_B \compose
  \varepsilon_B$, i.e.\ $\varepsilon_B$ is a $\calT$-homomorphism.
  Therefore, by uniqueness of a final morphism, we have
  $\fmorph{\calT}{\model}{\cdot} = \fmorph{\calT}{{\model_B}}{\cdot}
  \compose \, \varepsilon_B$.  In particular, with respect
  to~$\model$, this implies $\fmorph{\calT}{\model}{x_1} =
  \fmorph{\calT}{\model}{x_2}$ since $(x_1,x_2)\in B$ and so
  $(\varepsilon_B \, x_1) = (\varepsilon_B \, x_2)$.  Thus, $x_1
  \beq{\calT}{} x_2$.

  For the reverse---i.e.\ ${x_1 \beq{\calT}{} x_2 }$ implies ${x_1
    \ebisim x_2}$---assume $x_1 \beq{\calT}{} x_2$, i.e.
  $\fmorph{\calT}{\model}{x_1} = \fmorph{\calT}{\model}{x_2}$, for
  $x_1,x_2 \in X$. Define the relation $R \subseteq X \times X$
  such that $(x_1,x_2) \in R$ iff $\fmorph{\calT}{\model}{x_1} =
  \fmorph{\calT}{\model}{x_2}$. We first show that $R$ is an
  $\eta$-bisimulation. Suppose $(s',s'') \in R$ and recall that
  $\fmorph{\calT}{\model}{\cdot} : (X,\eta) \to (\Omega,\omega)$ is a
  $\calT$-homomorphism. For what concerns the first condition of
  Definition~\ref{def:ETB} we have
  \begin{displaymath}
    \def\arraystretch{1.0}
    \begin{array}{rcll}
      (\eta \mkern1mu s')_1 & = &
      ((\calT \, \fmorph{\calT}{}{\cdot})(\eta \,s'))_1
      & \text{(def.\ of $\calT$)}
      \\ & = &
      (((\calT \,\fmorph{\calT}{}{\cdot}) \compose \eta) \, s')_1
      \\ & = &
      (( \omega \compose \fmorph{\calT}{}{\cdot})s')_1
      & \text{($\fmorph{\calT}{}{\cdot} : (X,\eta) \to
        (\Omega,\omega)$ is homomorphism.)} 
      \\ & = &
      (\omega \, \fmorph{\calT}{}{s'})_1
      \\ & = &
      (\omega \, \fmorph{\calT}{}{s''})_1
      & \text{($\fmorph{\calT}{}{s'} =
        \fmorph{\calT}{}{s''}$ since $(s',s'') \in R$)} 
      \\ & = &
      ((\omega \compose \fmorph{\calT}{}{\cdot})s'')_1
      \\ & = &
      (((\calT \,\fmorph{\calT}{}{\cdot}) \compose \eta) \, s'')_1
      & \text{($(\calT \,\fmorph{\calT}{}{\cdot}) \compose \eta =
        \omega \compose \fmorph{\calT}{}{\cdot}$)} 
      \\ & = &
      ((\calT \, \fmorph{\calT}{}{\cdot})(\eta \,s''))_1
      \\ & = &
      (\eta \mkern1mu s'')_1
      & \text{(def.\ of $\calT$).}
    \end{array}
    \def\arraystretch{1.0}
  \end{displaymath}
  For what concerns the second condition of Definition~\ref{def:ETB}
  we have, for $h \in \SET{2,3}$ and all $C \in X/R$, that
\begin{displaymath}
  \def\arraystretch{1.3}
  \begin{array}{rcll}
    \multicolumn{4}{l}{(\eta \mkern1mu s')_h \cap C \neq \emptyset}
    \\ 
    \mbox{\qquad} & \Leftrightarrow &
    (\eta \, s')_h \cap \fmorph{\calT}{-1}{w} \neq \emptyset \mbox{\qquad} &
    \\ & & 
    \multicolumn{2}{l}{\qquad \text{(def.\ of ${\fmorph{\calT}{-1}{\cdot}}$;
      def.\ of~$R$; $w = \fmorph{\calT}{}{t}$ for all $t \in C$)}}
    \\ & \Leftrightarrow &
    w \in (\omega {\fmorph{\calT}{}{s'}})_h
    & \text{(by Lemma~\ref{lemma:aux} below)}
    \\ & \Leftrightarrow &
    w \in (\omega {\fmorph{\calT}{}{s''}})_h
    & \text{($ \fmorph{\calT}{}{s'} =
      \fmorph{\calT}{}{s''}$ since $(s',s'') \in R$)} 
    \\ & \Leftrightarrow &
    (\eta \mkern1mu s'')_h \cap \fmorph{\calT}{-1}{w} \neq \emptyset
    & \text{(by Lemma~\ref{lemma:aux} below)}
    \\ & \Leftrightarrow &
    (\eta \mkern1mu s'')_h \cap C \neq \emptyset
    \\ & & 
    \multicolumn{2}{l}{\qquad \text{(def.\ of ${\fmorph{\calT}{-1}{\cdot}}$;
        def.\ of~$R$; $w = \fmorph{\calT}{}{t}$ for all $t \in C$)}.}
  \end{array}
  \def\arraystretch{1.0}
\end{displaymath}
Since both conditions of Definition~\ref{def:ETB} are fulfilled, 
$R$ is an $\eta$-bisimulation relation and hence, since $(x_1,x_2)\in R$, 
we get ${x_1 \ebisim x_2}$. This completes the proof. \pfend
\end{proof}

\noindent
In the proof of Theorem~\ref{theo:etabisbeh} we have made use of the
following result.

\begin{lemma}
  \label{lemma:aux}
  For $h \in \SET{2,3}$, all $s \in X, w \in \Omega$ we have that $w
  \in (\omega \, {\fmorph{\calT}{}{s}})_h$ if and only if $(\eta \,
  s)_h \cap \fmorph{\calT}{-1}{w} \not=\emptyset$.
\end{lemma}

\begin{proof}
See Appendix \ref{appendix}.
\end{proof}

%

\noindent
From Theorem~\ref{theo:etabisbeh} and
Lemma~\ref{lemma:ebisISaltcmbisim} we get complete correspondence of
behavioural equivalence and bisimilarity for quasi-discrete closure
models .

\begin{corollary}
  \label{coro:bisbeh}
  Given a quasi-discrete closure model $\model = ((X,\closurename{R}),
  \pevalname)$ based on~$R$, for all $x_1, x_2 \in X$ it holds that
  $x_1 \altcmbisim^{\model} x_2$ iff $x_1 \beq{\calT}{\model} x_2$.
\end{corollary}

\begin{example}
  \label{exa:ExaTre}
  With reference to Example~\ref{exa:ExaDue}, we see that the minimal
  coalgebra for the models $\model_a, \model_b$ and~$\model_c$ is
  represented by $(\SET{\tau_1,\tau_2},\mu)$ where $\mu \mkern1mu
  \tau_1 = (\mathtt{blue}, \SET{\tau_1,\tau_2}, \SET{\tau_1,\tau_2})$,
  and $\mu \mkern1mu \tau_2 = (\mathtt{red}, \SET{\tau_1,\tau_2},
  \SET{\tau_1,\tau_2})$. The quotient morphisms are the obvious ones;
  for instance, letting model~$\model_c$ be represented by the
  coalgebra\footnote{For the sake of readability, here we use the same
    names~$c_{ij}$ for the elements of the carrier of the relevant
    coalgebra as those we used for defining the model $\model_c$,
    although everything is to be intended up to isomorphisms.}
  $(X_c, \eta_c)$, 
  $h_c :  (X_c, \eta_c) \to ((\tau_1,\tau_2),\mu)$ maps $c_{22},
  c_{23}, c_{32}$, and $c_{33}$ to~$\tau_2$, and all other elements
  to~$\tau_1$. 
  The minimal coalgebra for~$\model_d$ is
  $(\SET{\tau_3,\tau_4,\tau_5},\nu)$, where
  $\nu \mkern1mu \tau_3 = (\mathtt{red}, \SET{\tau_3,\tau_4},
  \SET{\tau_3,\tau_4})$, 
  $\nu \mkern1mu \tau_4 = (\mathtt{red}, \SET{\tau_3,\tau_4,\tau_5},
  \SET{\tau_3,\tau_4,\tau_5})$, and 
  $\nu \mkern1mu \tau_5 = (\mathtt{blue},\SET{\tau_4,\tau_5},
  \SET{\tau_4,\tau_5})$. The minimal models obtained using \texttt{MiniLogicA} are reported in Figure~\ref{fig:minimized-examples}. 
  Note that $h_d:(X_d, \eta_d) \to ((\tau_3,\tau_4,\tau_5),\nu)$ maps
  $d_{33}$ to~$\tau_3$, the elements of $\lbrace \, d_{ij} \mid 2
  \leqslant i,j \leqslant 4 \, \rbrace \setminus \SET{d_{33}}$
  to~$\tau_4$, and all the other elements to~$\tau_5$.
  Finally, consider the union~$\model_{abcd}$ of~model $\model_{abc}$
  and model~$\model_d$.  In this case, the minimal coalgebra is
  (isomorphic to) $(\SET{\tau_1, \tau_2, \tau_3, \tau_4, \tau_5},\mu+\nu)$
  where $(\mu+\nu) \mkern1mu \tau = \mu \mkern1mu \tau$ if $\tau =
  \tau_1, \tau_2$ and $(\mu+\nu) \mkern1mu \tau = \nu \mkern1mu \tau$
  if~$\tau = \tau_3, \tau_4, \tau_5$.
\end{example}

\begin{figure}
  \centering
  \fbox{
    \subfloat[][]
    { \includegraphics[trim=0pt 15pt 15pt 20pt,clip,height=50pt]{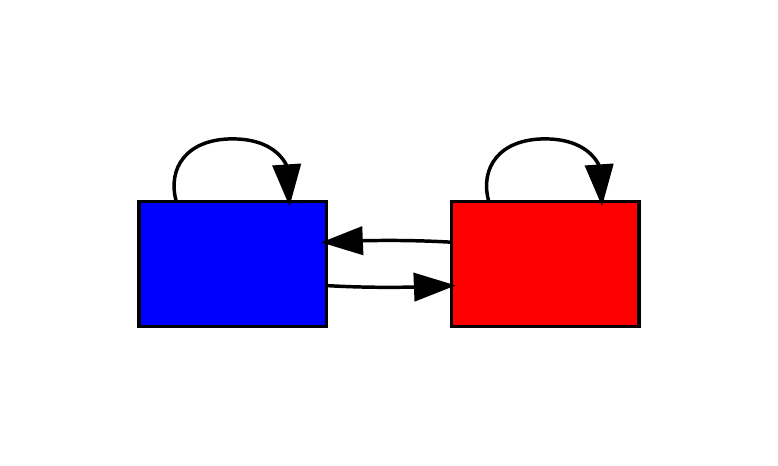}
      \label{fig:minimized-examples-a}
    }
    \subfloat[][]
    { \includegraphics[trim=0 15pt 15pt 20pt,clip,height=50pt]{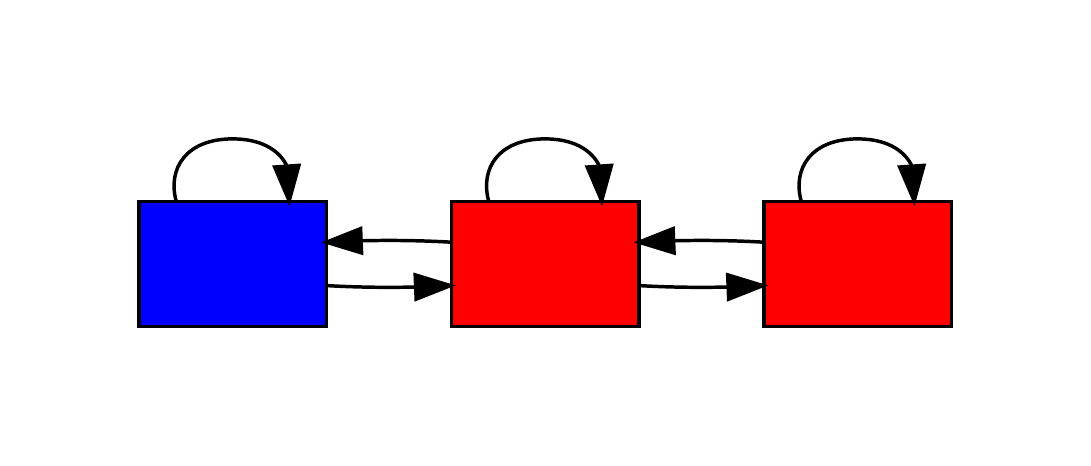}
      \label{fig:minimized-examples-b}
    }
  }
  \caption{\label{fig:minimized-examples}Minimal models obtained by running {\tt MiniLogicA} on the images of Figure~\ref{fig:ExaDue}. Models $\model_a, \model_b$ and~$\model_c$ are all equivalent to the model on the left, whereas $\model_d$ is equivalent to the model on the right.}
\end{figure}

%% file: SLCSEq.tex
\section{$\SLCS\ $ and logical equivalence}
\label{sec:SLCSEq}

We use the following version of the logic~$\SLCS$ for a given set
atomic propositions~$AP$. 

\begin{equation}
  \label{def:syntax}
  \form  ::=  p  \msep  \lneg \, \form  \msep  \form \, \lor \, \form  \msep 
  \lrcs{\form}{\form}   \msep  \lrcd{\form}{\form} 
\end{equation}
Satisfaction $\model, x \models \form$ of a formula~$\form$ at point
$x \in X$ in a quasi-discrete closure model $\model = ((X,\closurename{R}), \peval)$ is defined in Figure~\ref{def:satisfaction} by induction on the structure of formulas. 
 
\begin{figure}
\centering
\emph{%
\fbox{%
\begin{math}
\begin{array}{r c l c l c l l}
\model,x & \models  & p \in P & \Leftrightarrow & x  \in \peval p \\
\model,x & \models  & \lneg \,\form & \Leftrightarrow & \model,x  \models \form \mbox{ does not hold}\\
\model,x & \models  & \form_1\, \lor \,\form_2 & \Leftrightarrow &
\model,x  \models \form_1 \mbox{ or } \model,x  \models \form_2\\
\model,x & \models  & \lrcs{\form_1}{\form_2} & \Leftrightarrow &
\mbox{there exists a path } \pi  \mbox{ and an index } \ell \mbox{
  such that 
}\\
 & &  & & 
\mbox{\hspace{0.1in}}\pi(0)=x \, \mbox{ and } \,  \model, \pi(\ell) \models \form_1 \mbox{ and }\\
 & &  & & 
\mbox{\hspace{0.1in}}
\model, \pi(j) \models \form_2, \mbox{ for all } j \mbox{ with }
0<j<\ell \\
\model,x & \models  & \lrcd{\form_1}{\form_2} & \Leftrightarrow &
\mbox{there exists a path } \pi  \mbox{ and an index } \ell \mbox{
  such that
} \\
 & &  & & 
\mbox{\hspace{0.1in}}\pi(\ell)=x \, \mbox{ and } \,  \model, \pi(0) \models \form_1 \mbox{ and }\\
 & &  & & 
\mbox{\hspace{0.1in}}
\model, \pi(j) \models \form_2, \mbox{ for all } j \mbox{ with } 0<j<\ell\\
\end{array}
\end{math}
} 
} 
\caption{Definition of the satisfaction relation\label{def:satisfaction}}
\end{figure}

\noindent
Some useful abbreviations are defined in
Figure~\ref{def:abbreviations}. 
The operator~$\lnear$ is the \emph{near} operator, namely the logical
counterpart of the closure function of closure spaces: $x$~satisfies
$\lnear \form$ iff it is ``close'' to~$\form$, as defined
in~\cite{CLLM16}\footnote{In~\cite{CLLM16} the following definition
  has been used: $\model, x \models \lnear \form \Leftrightarrow x \in
  \closure{R}{\ZET{y}{\model, y \models \form}}$; it is easy to show
  that $\model, x \models \lnear \form$ if and only if $\model, x
  \models \lrcd{\form}{\lfalse}$.}.
The operator~$\lsurr$ is the \emph{surrounded}
\footnote{Named ``spatial until'' and denoted by ``${\cal U}$''
  in~\cite{CLLM14}.} operator: a point $x$ satisfies $\form_1 \,
\lsurr \, \form_2$ if it satisfies~$\form_1$ and no path starting
at~$x$ can reach any point satisfying $\neg \, \form_1$ without first
passing by a point satisfying~$\form_2$, i.e.\ $x$~lays in an area
that satisfies~$\form_1$ and that is surrounded by points satisfying
$\form_2$.
The operator~$\lprop$ is the \emph{propagation} operator introduced
in~\cite{CLLM16}: $x$~satisfies $\form_1 \, \lprop \, \form_2$ if it
satisfies~$\form_2$ and it is reachable from a point
satisfying~$\form_1$ via a path such that all of its points, except
possibly the starting point, satisfy~$\form_2$. For more derived
operators the reader is referred to~\cite{CLLM16}. 

\begin{example} Consider Example~\ref{exa:ExaDue}. All the red points in Figure~\ref{fig:ExaDue}(a-c) satisfy $\lnear\, {\tt blue}$. The middle point in Figure~\ref{fig:ExaDue}(d) (point $d_{33}$) does not satisfy such formula, although it satisfies $\lrcs{\tt blue}{\tt red}$.
\end{example}

\begin{figure} 
\centering
\emph{%
\fbox{%
\begin{math}
\begin{array}{l c l}
  \form_1 \, \land \, \form_2 & \equiv & \lneg (\lneg \form_1 \, \lor
  \, \lneg \form_2) \\
  \lfalse & \equiv & p \, \land \, \lneg p \\
  \ltrue & \equiv & \lneg\lfalse \\
  \lnear \form & \equiv & \lrcd{\form}{\lfalse} \\
  \form_1 \, \lsurr \, \form_2 & \equiv & \form_1 \, \land \,
  \lneg(\lrcs{\lneg(\form_1 \, \lor \, \form_2)}{\lneg \form_2}) \\
  \form_1 \, \lprop \, \form_2 & \equiv & \form_2 \land
  \lrcd{\form_1}{\form_2} 
\end{array}
\end{math}
} 
} 
\caption{Derived operators~\label{def:abbreviations}}
\end{figure}

\begin{definition}
  \label{def:SLCSEQ}
  The $\SLCS $\ equivalence relation with respect to model~$\model$,
  namely $\slcseq^{\model} \subseteq X \times X$, is defined as
  follows: $x_1 \slcseq^{\model} x_2$ iff for all $\SLCS$
  formulas~$\form$ we have that $\model, x_1 \models \form
  \Leftrightarrow \model, x_2 \models \form$.
\end{definition}

\noindent
In the following, for the sake of notational simplicity, we will write
$\slcseq$ instead of $\slcseq^{\model}$ whenever this cannot cause
confusion.

\begin{lemma}
  \label{lemma:slcseqsubcsbisim}
  If a quasi-discrete closure model~$\model$ is finitely closed and
  finitely backward closed, then $\slcseq \, \subseteq\,  \altcmbisim$.
\end{lemma}

\begin{proof}
See Appendix~\ref{appendix}.
\end{proof}

\begin{lemma}\label{lemma:csbisimsubslcseq}
If $\model$ is a quasi-discrete closure model, then $\altcmbisim\, \subseteq\,  \slcseq$.
\end{lemma}

\begin{proof}
See Appendix~\ref{appendix}.
\end{proof}


\noindent
Lemma~\ref{lemma:slcseqsubcsbisim} and
Lemma~\ref{lemma:csbisimsubslcseq} bring the following result.

\begin{theorem}
  \label{correspondence}
  If closure model $\model = ((X,\closurename{}), \peval)$ is
  quasi-discrete, finitely closed, and finitely backward closed,
  then for all $x_1,x_2 \in X$, $ x_1 \, \altcmbisim\, x_2$ iff we
  have $x_1 \, \slcseq\, x_2$.
\end{theorem}

\begin{example} Let us provide formulas that uniquely characterise each of the points in Figure~\ref{fig:minimized-examples}, proving that there are no different, bisimilar points in that figure, when it is considered as a single model by taking the disjoint union of the models (a) and (b). 
Let $\Phi_1 = {\tt blue} \land \lnear {\tt red}$, $\Phi_2 = {\tt red} \land \lnear {\tt blue}$, $\Phi_3 =  {\tt red} \land \lnot \lnear {\tt blue}$. 
The blue point in Figure~\ref{fig:minimized-examples}(a) is the only point satisfying $\Phi_1 \land \lnot \lrcs{\Phi_3}{\ltrue}$. The red point in Figure~\ref{fig:minimized-examples}(a) is the only point satisfying $\Phi_2 \land \lnot \lnear \Phi_3$. The blue point in Figure~\ref{fig:minimized-examples}(b) is the only point satisfying $\Phi_1 \land \lrcs{\Phi_3}{\ltrue}$. The middle red point in Figure~\ref{fig:minimized-examples}(b) is the only point satisfying $\Phi_2 \land \lnear \Phi_3$. The rightmost red point in Figure~\ref{fig:minimized-examples}(b) is the only point satisfying $\Phi_3$.
\end{example}

\noindent
We close this section with a stronger version of
Lemma~\ref{lemma:slcseqsubcsbisim}, and consequently of
Theorem~\ref{correspondence}.  Let us consider the sub-logic $\SMCSm$
of~$\SLCS$ given by
\begin{equation}\label{eq:sublogic}
  \form  ::=  p  \msep  \lneg \, \form  \msep  \form \, \lor \, \form  \msep
  \lrcs{\form}{\lfalse} \msep \lrcd{\form}{\lfalse}
\end{equation}
and let $x_1 \, \slcsmeq\, x_2$ denote the logical equivalence with
respect to the sub-logic~$\SMCSm$.

Lemma~\ref{lemma:restrict} below lays the basis for showing that for
two points $x_1$ and~$x_2$ with $x_1 \, \slcsmeq\, x_2$ also holds
that $x_1 \, \slcseq\, x_2$, i.e.\ using the full version
$\lrcs{\form_1}{\form_2}$ and $\lrcd{\form_1}{\form_2}$ of the
$\lrcsname$ and~$\lrcdname$ operators does not add discriminatory
power with respect to using the restricted versions
$\lrcs{\form}{\lfalse}$ and $\lrcd{\form}{\lfalse}$.

\begin{lemma}
  \label{lemma:restrict}
  Formulas $\form_1,\form_2$ of~$\SMCSm$ of Equation~\ref{eq:sublogic}
  satisfy the following.
  \begin{enumerate}
  \item If $\model, x_1 \models \lrcs{\form_1}{\form_2}$ and $\model,
    x_2 \not\models \lrcs{\form_1}{\form_2}$ then there exists
    $\Lambda_{\form_1,\form_2}$ in the language of
    Equation~\ref{eq:sublogic} such that $\model, x_1 \models
    \Lambda_{\form_1,\form_2} $ and $\model, x_2 \not\models
    \Lambda_{\form_1,\form_2} $.
  \item If $\model, x_1 \models \lrcd{\form_1}{\form_2}$ and $\model,
    x_2 \not\models \lrcd{\form_1}{\form_2}$ then there exists
    $\Lambda_{\form_1,\form_2}$ in the language of
    Equation~\ref{eq:sublogic} such that $\model, x_1 \models
    \Lambda_{\form_1,\form_2} $ and $\model, x_2 \not\models
    \Lambda_{\form_1,\form_2} $.
\end{enumerate}
\end{lemma}

\begin{proof}
See Appendix~\ref{appendix}.
\end{proof}

\noindent
Using the above lemma, one can then prove the following result.

\begin{theorem}
  \label{correspondencem} 
  If closure model $\model = ((X,\closurename{}), \peval)$ is
  quasi-discrete, finitely closing, and finitely backward closing,
  then for all $x_1,x_2 \in X$ we have $ x_1 \, \altcmbisim\, x_2$ iff
  $x_1 \, \slcsmeq\, x_2$.
\end{theorem}

\begin{remark}
  \label{rem:BaFLog}
  With reference to Remark~\ref{rem:BaFbis} it is easy to see that
  while $\model, x_1 \models \lrcd{p}{\lfalse}$, we have $\model, x_2
  \not\models \lrcd{p}{\lfalse}$. 
\end{remark}


%% file: Tool.tex
\section{A tool for spatial minimization}
\label{sec:tool}

One of the major advantages of defining bisimilarity coalgebraically
is the availability of the partition refinement algorithm, sometimes
referred to as \emph{iteration along the final sequence} (see
e.g.~\cite{AdamekBHKMS12}). In the category $\Set$, the formulation of
the algorithm is particularly simple and quite similar to classical
results such as~\cite{HOPCROFT1971189}. In
Algorithm~\ref{fig:coalgebraic-partition-refinement}, we illustrate
the algorithm. For $q$ a function, we let $\mathit{ker}(q)$ be its
\emph{kernel}, namely the partition of the domain induced by~$q$.

\begin{algorithm}
    \textbf{function} \texttt{minimizeRec} ($\eta : X \to \calF X$,$q :
    X \to \calF^k \{*\}$) 

    \quad \textbf{let} $q' = (\calF q) \circ \eta$

    \quad \textbf{if} ($\mathit{ker}(q) = \mathit{ker}(q')$)
    \textbf{then} 

    \qquad \textbf{return}~$q$ 
    
    \quad \textbf{else} 

    \qquad \textbf{return} \texttt{minimizeRec($\eta$,$q'$)}

    \textbf{function} \texttt{minimize}($\eta : X \to \calF X$)

    \quad \textbf{return} \texttt{minimizeRec($\mathit{\eta}$,$\lambda
      x . * $)} 
\caption{The coalgebraic partition refinement algorithm in $\Set$.}
\label{fig:coalgebraic-partition-refinement}
\label{alg:minimize}
\end{algorithm}

\noindent The function \texttt{minimize} accepts as input a $\calF$-coalgebra $\eta$ and returns the bisimilarity quotient of its carrier set. Minimization is implemented via the function \texttt{minimizeRec}, which accepts as input the coalgebra map $\eta$, and a surjective function $q$, whose kernel is a partition of the carrier set. Such function is initialised to $\lambda x . *$, where $*$ is the only element of the singleton $\{ * \}$, that is, the algorithm starts by assuming that all the elements of the carrier are bisimilar. The algorithm then applies one \emph{refinement} step, by applying the functor $\calF$ to $q$ and composing the result with $\eta$; this yields a new function $q' = (\calF q) \circ \eta$. Note that such function is ``almost always'' surjective\footnote{Function $q'$ may actually fail to be surjective when the carrier is empty. All $\Set$ functors preserve epimorphisms from non-empty sets. If the carrier is empty then so are both $q$ and $q'$, therefore the algorithm terminates in one step.}. Intuitively, at each iteration,  function $q'$ is obtained from $q$ by splitting the partitions induced by $q$ according to the ``observations'' that are obtained ``in one more step'' from $\eta$.     If $q$ and $(\calF q) \circ \eta$ represent the same partition -- that is, the two functions have the same kernel -- the algorithm returns $q$, which denotes the coarsest partition that does not identify non-bisimilar states; otherwise, the procedure is iterated. Termination is guaranteed on finite models as for each finite model, there are  only a finite number of partitions.

Algorithm \ref{alg:minimize}, instantiated using the functor $\calT$ of
Section~\ref{sec:CMCoalgebraically}, has been implemented in a multi-platform
tool called \textbf{MiniLogicA}, which is available for the major operating
systems at \url{https://github.com/vincenzoml/MiniLogicA} under a permissive
open source license. The tool is implemented in the language F\#\footnote{See
\url{https://fsharp.org}.}. The tool can load arbitrary (possibly directed)
graphs, with explicit labelling of nodes with atomic propositions. Such labelled
graphs are interpreted as quasi-discrete closure models. Additionally, the tool
can load digital images, that are interpreted as symmetric, grid-shaped graphs,
therefore as quasi-discrete models. More precisely, each pixel is interpreted as
a node of a graph, and atomic propositions are derived from RGB colour
components, whereas connectivity is derived from the union of the relations
between pixels ``have an edge in common'' and ``have a vertex in common'' (in 2
dimensions, this corresponds to the classical \emph{orthodiagonal connectivity},
that is, each non-border pixel is connected to 8 other pixels). The tool
currently supports 2D images, but support of the same formats as \voxlogica{} is
planned. The tool outputs graphs in the graphviz format\footnote{See
\url{https://www.graphviz.org}.}, with labels using atomic propositions, or
colours according to the pixel colours of the input in the case of images.

%% file: Generalisation.tex
\section{Extension to Generic Closure Spaces}
\label{sec:generalization}

In this section we provide first a set-theoretic and next a
coalgebraic notion of bisimilarity for closure models that aren't
necessarily quasi-discrete, and we prove that both coincide with
logical equivalence as induced by an infinitary modal logic, here
called \INFLOG, which, compared to $\SLCS$, does not include
reachability operators. Instead, $\lnear \form$~is the basic
operator---endowed with the classical closure semantics. Also,
infinitary conjunction is allowed.

\begin{equation}
  \label{def:infsyntax}
  \form  ::=  p  \msep \lneg \, \form \msep \textstyle{\linfand{i \in I}}
  \: \form_i \msep \lnear \form 
\end{equation}
where $p \in AP$, and $I$~is a set. 

For a closure model $((X,\closurename{}),\calV)$ we have, as expected,
$\model, x \models p$ $\Leftrightarrow$ $x \in \calV p$, $\model, x
\models \lnear \form$ $\Leftrightarrow$ $x \in \closure{} {\ZET{y
    \!}{\! \model, y \models \form}}$, $\model,x \models \lnot \form$
  $\Leftrightarrow$ $\model,x \not \models \form$, and finally
  $\model,x \models \linfand{i \in I} \form_i$ $\Leftrightarrow$ $\model,x
  \models \form_i$ for all~$i \in I$.

\begin{definition}
  \label{def:LEQ}
  The equivalence relation $\inflogeq^{\model} \subseteq X \times X$
  is defined by $x_1 \inflogeq^{\model} x_2$ iff for all
  \INFLOG\ formulas $\form$ we have that $\model, x_1 \models \form
  \Leftrightarrow \model, x_2 \models \form$.
\end{definition}

\noindent
In the sequel, $\inflogeq^{\model}$ will be often abbreviated
by~$\inflogeq$.  The following definition extends the notion of
bisimulation for topological spaces (see~\cite{BeB07}, for example) to
general closure models.

\begin{definition}
  \label{def:GCMB}
  Given a closure model $\model = ((X,\closurename{}), \pevalname)$, a
  non-empty equivalence relation $B \subseteq X \times X$ is called a
  \emph{bisimulation relation} if, for all $x_1,x_2 \in X$ such that
  $(x_1,x_2) \in B$, the next two conditions are satisfied.
  \begin{enumerate}
  \item $\invpeval{x_1} = \invpeval{x_2}$.
  \item For all $X_1 \subseteq X$ such that $x_1 \in
    \interior{}{X_1}$, there is $X_2 \subseteq X$ such that $x_2 \in
    \interior{}{X_2}$ and, reversely, for all $x_2' \in X_2$ there
    exists $x_1' \in X_1$ such that $(x_1',x_2') \in B$.
  \end{enumerate}
  We say that $x_1$ and~$x_2$ are {\em bisimilar}, notation $x_1
  \gcmbisim x_2$, if there exists a bisimulation relation~$B$ for~$X$
  such that $(x_1,x_2) \in B$.
\end{definition}

\begin{remark}
  The definition of~\cite{BeB07} (given for topological models)
  differs from the definition above in that the sets~$X_{i}$ are
  required to be \emph{open} neighbourhoods. In topology, a subset~$S$
  is an open neighbourhood of a point~$x$ whenever there is an open
  set~$O$ with $x \in O \subseteq S$, or, equivalently, $x \in
  \interior{}{S}$. Therefore, in a topological space,
  Definition~\ref{def:GCMB} coincides with the one
  of~\cite{BeB07}. However, in general closure models, this is
  different. For instance consider a graph with three nodes $a,b,c$,
  and relation $R = \{(a,b),(b,c)\}$. Let $S = \{b,c\}$. We have
  $\interior{}{S} = S \setminus \closure{}{\overline S} = S \setminus
  \closure{}{\{a\}} = S \setminus \{a,b\} = \{c\} \neq S$, therefore
  $S$~is not open (see also~\cite{CLLM16}, Remark 2.19).  Similarly,
  $\{c\}$~is not open as $\interior{}{\{c\}} = \emptyset$. Thus
  $S$~does not include an open set containing~$c$. However, $c \in
  \interior{}{S}$.
\end{remark}

\noindent 
Below, we show that logical equivalence in~\INFLOG coincides with
bisimilarity from Definition~\ref{def:GCMB}.
The following two lemmas are required.

\begin{lemma}
  \label{lem:closure-cap-interior}
  For all $X_1,X_2 \subseteq X$, if $(\closure{}{X_1}) \cap
  (\interior{}{X_2}) \neq \emptyset$ then $X_1 \cap X_2 \neq
  \emptyset$.
\end{lemma}

\begin{proof}
  We prove that $X_1 \cap X_2 = \emptyset$ implies $(\closure{}{X_1})
  \cap (\interior{}{X_2}) = \emptyset$. Suppose $X_1 \cap X_2 =
  \emptyset$. Then $X_1 \subseteq \overline X_2$, thus
  $\closure{}{X_1} \subseteq \closure{}{\overline X_2}$. Since
  $\interior{}{X_2} = \overline{\closure{}{(\overline X_2)}}$, it
  follows that $(\interior{}{X_2}) \cap (\closure{}{X_1}) =
  \emptyset$. \qed
\end{proof}

\begin{lemma}
  \label{lem:interior-implies-closure}
  For all $S \subseteq X$ and $y \in X$, if for all $C \subseteq X$ it
  holds that $y \in \interior{}{C}$ implies $C \cap S \neq \emptyset$,
  then $y \in \closure{}{S}$.
\end{lemma}

\begin{proof}
  By contradiction.  Suppose $y \notin \closure{}{S}$ under the
  hypothesis of the lemma. Then $y \in \overline{\closure{}{S}}$,
  i.e.\ $y \in \overline{\overline{\interior{}{(\overline{S})}}} =
  \interior{}{(\overline{S})}$. But then, by the hypothesis, taking $C
  = \overline{S}$ since $\overline{S} \subseteq X$, we would have that
  $\overline{S} \cap S \neq \emptyset$. \qed
\end{proof}

\noindent
With the two lemmas in place, we are in a position to prove the next
results. 

\begin{theorem}
  Given a closure model $\model = ((X,\closurename{}), \pevalname)$, any
  bisimulation~$B$ according to Definition~\ref{def:GCMB} is included
  in~$\inflogeq^\model$.
\end{theorem}

\begin{proof}
  By induction on the structure of~$\form$.
  $\form = \lnear \form'$. Suppose $B$ is a bisimulation, $(x,y) \in
  B$ and, without loss of generality, $\model,x \not \models \form$
  and $\model, y \models \form$. Let $F \subseteq X$ be the set of
  points satisfying~$\form'$. We have $y \in \closure{}{F}$ and $x \in
  \overline{\closure{}{F}} =
  \overline{\overline{\interior{}{(\overline F)}}} =
  \interior{}{(\overline F)}$. Let $X_1 = \overline F$. By $x \in
  \interior{}{X_1}$, let $X_2$ be chosen according to
  Definition~\ref{def:GCMB}, with $y \in \interior{}{X_2}$. By
  Lemma~\ref{lem:closure-cap-interior} we have $F \cap X_2 \neq
  \emptyset$, since $y \in (\closure{}{F}) \cap
  (\interior{}{X_2})$. Let $y' \in F \cap X_2$. We have $\model, y'
  \models \Phi'$, since $y' \in F$. Since $B$ is a bisimulation
  according to Definition~\ref{def:GCMB}, there is $x' \in X_1$ with
  $(x',y')\in B$. By the induction hypothesis $\model, x' \models
  \form'$, thus $x' \in F$, which contradicts $x' \in X_1 = \overline
  F$.
\end{proof}

\begin{theorem}\label{theo:inflogeq:bisim}
  Given model $\model=((X,\closurename{}), \pevalname)$,
  $\inflogeq^\model$ is a bisimulation according to
  Definition~\ref{def:GCMB}.
\end{theorem}

\begin{proof}
  Suppose $x \inflogeq y$. Let $X_1$ be such that $x \in
  \interior{}{X_1}$. Suppose there is no~$X_2$ respecting the
  conditions of Definition~\ref{def:GCMB}. Then, either there is no $C
  \subseteq X$ such that $y \in \interior{}{C}$ or for each such~$C$
  there is $y_{\mkern1mu C} \in C$ such that $x'\inflogeq y_{\mkern1mu
    C}$ for \emph{no} $x' \in X_1$. In the first case we would have
  that, for all $C \subseteq X$, $y \notin
  \overline{\closure{}{(\overline{C})}}$, i.e.\ $y \in
  \closure{}{(\overline{C})}$. This would imply in turn that $y \in
  \closure{}{(\overline{X})} = \emptyset$, which is absurd. In the
  second case, let $S$ be the set of all the~$y_{\mkern1mu C}$ as
  above. We have $y \in \closure{}{S}$ by
  Lemma~\ref{lem:interior-implies-closure}. For each $a \in X_1$ and
  $s \in S$, $a$ and~$s$ are not logically equivalent: let
  $\form_{(a,s)}$ be a formula such that $\model,a \not \models
  \form_{(a,s)}$ and $\model,s \models \form_{(a,s)}$. Let $\form =
  \bigwedge_{\, s} \lnot \bigwedge_{\, a} \form_{(a,s)}$. We have that
  $\model, x' \models \form$ for all $x' \in X_1 $ and $\model, y'
  \not \models \form$ for all $y' \in S$. To see the latter, observe
  that $\lnot \Phi = \bigvee_{\! s} \bigwedge_{\, a}
  \Phi_{(a,s)}$. For each $a$, each $y' \in S$ satisfies at least
  $\bigwedge_{\, a} \form_{(a,y')}$. Thus, we have a formula~$\Phi$ with
  $X_1 \subseteq F = \{ \, z \in X \mid z \models \Phi \, \}$ and $S
  \subseteq \overline F$. By $x \in \interior{}{X_1}$ and monotonicity
  of interior, we have $x \in \interior{}{F}$, thus $\model, x \models
  \lnot \lnear (\lnot \form)$. On the other hand, by $y \in
  \closure{}{S}$ and monotonicity of closure, we have $y \in
  \closure{}{(\overline F)}$, thus $\model, y \models \lnear (\lnot
  \form)$, contradicting the hypothesis $x \inflogeq y$.  \qed
\end{proof}

\noindent
The characterisation given by Definition~\ref{def:GCMB} has the merit of
extending the existing topological definition to closure spaces. However, in the
setting of this paper it is worthwhile to investigate also a coalgebraic
definition, which we do in the remainder of this section.  Since our main objective is to characterise logical equivalence,
we will not define \emph{frames}, but just \emph{models}, which we will call
\emph{closure coalgebras}.

\begin{definition}\label{def:closureCoalgebra}
  A \emph{closure coalgebra} is a coalgebra for the \emph{closure
    functor} $\closurefunctor X = \pws{(AP)} \times \pws{(\pws{X})}$,
  where $\pws{-}$ is the covariant powerset functor. The action of the functor on arrows maps $f : X \to Y$ to 
  $\closurefunctor f: (\pws{(AP)} \times \pws{(\pws{X})}) \to (\pws{(AP)} \times \pws{(\pws{Y})})$ such that 
  $\closurefunctor f (v, S)=(v,\{(\pws f) A | A \in S\})$.
\end{definition}

\noindent
We note in passing
that a general coalgebraic treatment of modal logics -- even the non-normal ones
-- can be done starting from \emph{neighbourhood frames}~\cite{HKP09}, employing
coalgebras for the functor $2^{2^{-}}$ (where $2^{-}$ is the contravariant power
set functor). Our definition is similar, but in contrast we employ the
\emph{covariant} powerset functor~$\pws{}$, which we find particularly
profitable, as the obtained theory is akin\footnote{In order to make Definition~\ref{def:closureCoalgebra} a proper generalisation of Definition~\ref{def:simplerfunctor}, one needs to identify the correct notion of \emph{path} for closure coalgebras (more on this in Section~\ref{sec:discussion}).} to the developments of Section~\ref{sec:CMCoalgebraically}.
The remainder of this section is aimed at determining a correspondence
between closure models, closure coalgebras, and their quotients.

\begin{definition}
  \label{def:closureCoalgebraOfAModel}
  Given a closure model $((X,\closurename{}),\pevalname)$, define the
  coalgebra $\eta : X \to \closurefunctor X$ by $\eta(x) =
  ( \mkern1mu \invpevalname{}x,\{ \, A \subseteq X \mid x \in
  \closurename{} A \mkern1mu \})$.
\end{definition}

\noindent It is straightforward to check that if $f : \calX \to \calY$ is a
$\closurefunctor$-coalgebra homomorphism, and both $\calX$ and $\calY$ have been
obtained from closure models using
Definition~\ref{def:closureCoalgebraOfAModel}, then $f$ is a continuous function
in the sense of Definition~\ref{def:ContinuousFunction}. From now on, we shall
not rely on the existence of a final coalgebra, as this is not the case for the
(unbounded) powerset functor. However, we can employ maximal quotients instead,
for the purpose of this paper. Therefore, we will redefine behavioural
equivalence from Section~\ref{sec:Preliminaries}.

\begin{definition}
  \label{def:behaviouralEquivalenceNew}
  Given a set functor~$\calF$ and a $\calF$-coalgebra $\calX =
  (X,\alpha)$ with $\alpha : X \to \calF X$, the relation
  $\beq{\calF}{\calX}$, defined by 
  \smallskip \\
  \centerline{%
  \begin{math}
    x \beq{\calF}{\calX} y
    \Leftrightarrow 
    \exists \calY = (Y, \beta) . \exists f : \calX \to \calY . f(x) = f(y)
  \end{math}}
  \smallskip
  is called \emph{behavioural equivalence}.
\end{definition}

\noindent
In Definition~\ref{def:behaviouralEquivalenceNew} we use the word
\emph{equivalence}, but this should not be taken for granted, of
course. Clearly, $\beq{}{}$~is reflexive and symmetric, but
transitivity is in principle to be shown. However (see
\cite{Hughes2001}, Theorem 1.2.4) pushouts in a $\Set$-based category
of coalgebras exist and are computed in the base category, which
immediately yields transitivity of~$\beq{}{}$. It is also obvious that
when a final coalgebra exists, $\beq{}{}$~coincides with the kernel of
the final morphism from~$\calX$.

\begin{lemma}
  \label{lem:quotientsOfClosureCoalgebras}
Consider a model $\model = ((X,\closurename{X}),\pevalname)$ and $\calX =
(X,\eta)$ as in Definition \ref{def:closureCoalgebraOfAModel}. Let $\calY =
(Y,\theta)$ be a $\closurefunctor$-coalgebra. Let $f : \calX \twoheadrightarrow
\calY$ be a surjective coalgebra homomorphism. Define $\closure{Y}{(B \subseteq
Y)} = \{ \, y \in Y \mid B \in (\theta y)_2 \, \}$. Then $(Y,\closurename{Y})$
is a closure space.
\end{lemma}

\begin{proof}
  See Appendix~\ref{appendix}.
\end{proof}
  
\noindent
The proof of Theorem~\ref{thm:closureCoalgebraLogicalEquivalence} below 
requires the following lemma, whose proof crucially relies on the fact
that $A \subseteq B$ implies $\closurename{} \mkern-2mu A \subseteq
\closurename{}B$.

\begin{lemma}
  \label{lem:closure-of-equivalent-points}
  Let $f$ be the function mapping each element of~$X$ into its
  equivalence class up to~$\inflogeq$. Then it holds that $((x_1
  \inflogeq x_2) \land x_1 \in \closure{} \mkern-2mu A)$ implies $x_2
  \in \closurename{} f^{-1} (\pws f) A$, for all $x_1,x_2 \in X$ and
  $A \subseteq X$.
\end{lemma}

\begin{proof}
  See Appendix~\ref{appendix}.
\end{proof}
  
\noindent
With Lemma~\ref{lem:quotientsOfClosureCoalgebras} and
Lemma~\ref{lem:closure-of-equivalent-points} available, we arrive at
the following result.

\begin{theorem}
  \label{thm:closureCoalgebraLogicalEquivalence}
  Consider a closure model $\model = ((X,\closurename{}),\pevalname)$
  and $\calX = (X,\eta)$, with~$\eta$ as in
  Definition~\ref{def:closureCoalgebraOfAModel}. It holds that the
  relations $\inflogeq^{\model}$ and $\beq{\closurefunctor}{\calX}$
  coincide.
\end{theorem}

\begin{proof}
    First, let us prove that if we have $\model, x_1 \models \Phi
    \Leftrightarrow \model, x_2 \models \Phi$ for all $\form$ and $x_1,
    x_2 \in X$, then there are a coalgebra $\calY = (Y,\theta)$ and a
    coalgebra homomorphism $f :X \to Y$ with $f x_1 = f x_2$.  Let $Y$
    be the set of equivalence classes of~$X$ under
    $\inflogeq^{\model}$.  Let $f$ be the canonical map, mapping each
    $x \in X$ to its equivalence class~$[x]$ with respect
    to~$\inflogeq^{\model}$. Note that each element of~$Y$ is of the
    form~$f \mkern-1mu x$ for some~$x$. Define $tx = \{ \, (\pws f) A
    \mid A \in (\eta x)_2 \, \}$, and let $\theta(f x) = ((\eta
    \mkern1mu x)_1,t(x))$. Observe that such a definition makes $f$ a
    coalgebra homomorphism by construction, that is, $\theta \circ f =
    (\closurefunctor f) \circ \eta$. We need to show that the
    definition of~$\theta$ is independent from the representative~$x$,
    i.e.\ whenever $x_1 \inflogeq x_2$, we have $\theta (f x_1) =
    \theta (f x_2)$.  Indeed, it is obvious that $(\eta x_1)_1 = (\eta
    x_2)_1$, since by logical equivalence $x_1$ and~$x_2$ satisfy the
    same atomic propositions. We thus need to show that $t\, x_1 = t\,
    x_2$. All elements of~$t\, x_1$ are of the form $(\pws f) A$ with
    $x_1 \in \closure{X} \mkern-2mu A$. By
    Lemma~\ref{lem:closure-of-equivalent-points}, we then have $x_2
    \in \closurename{X}(f^{-1}(\pws f)A)$, thus $f^{-1}((\pws f)A) \in
    (\eta x_2)_2$ by definition of~$\eta$.  Therefore, $(\pws f)
    (f^{-1} (\pws f) \mkern-0mu A) \in t x_2$ by definition of~$t$,
    and since $(\pws f) f^{-1} ((\pws f) \mkern-0mu A) = (\pws f)
    \mkern-0mu A$, we obtain $t\, x_1 \subseteq t\, x_2$.  The same
    reasoning can also be used in the other direction, proving that
    the two sets are equal.
    
    Next, we shall prove that if $((X,\closurename X),\pevalname{})$
    is a closure model, with corresponding $\closurefunctor$-coalgebra
    $(X,\eta)$, $(\calY,\theta)$ is a
    $\closurefunctor$-coalgebra, $f : X \to Y$ is a coalgebra
    homomorphism, and $f x_1 = f x_2$, then we have that $\model, x_1
    \models \Phi \iff \model, x_2 \models \Phi$ for all~$\Phi$.  We
    will actually prove a slightly stronger statement, based upon
    Lemma~\ref{lem:quotientsOfClosureCoalgebras}. Given that the
    category of $\closurefunctor$-coalgebras has a epi-mono
    factorization system inherited from~$\Set$ (that is, each
    coalgebra homomorphism can be written as $m\circ e$ where $e$~is
    surjective and $m$~is injective), let us restrict, without loss of
    generality, to the case when $f$ is surjective.  By
    Lemma~\ref{lem:quotientsOfClosureCoalgebras}, there is a closure
    operator $\closurename Y$ such that $\model' = ((Y,\closurename
    Y),\theta_1)$ is a closure model. Therefore, we can also interpret
    formulas on points of~$Y$.  Once this is established, under the
    hypothesis that $f$~is a (surjective) homomorphism, we shall prove
    that for all $x \in X$, we have $\model, x \models \Phi \iff
    \model', fx \models \Phi$ for all~$\Phi$.  This entails the main
    thesis as follows: whenever $f x_1 = f x_2$, for all~$\Phi$, we
    have $\model, x_1 \models \Phi \iff \model',f x_1 \models \Phi
    \iff \model', f x_2 \models \Phi \iff \model, x_2 \models \Phi$.  The proof proceeds by induction
    on the structure of~$\Phi$. The relevant case is that for formulas
    of the form~$\lnear \Phi$. The proof of this case is split into two
    directions.  Below, for any $\Phi$, we denote by~$S^X_\Phi$ the
    set $\{ \, x \in X \mid \model, x \models \Phi \, \}$ and
    with~$S^Y_\Phi$ the set $\{ \, y \in Y \mid \model', y \models
    \Phi \, \}$.
    
    ($\Rightarrow$) If $\model, x \models \lnear \Phi$, then $x \in
    \closurename X S^X_\Phi$ by definition of satisfaction, hence
    $S^X_\Phi \in (\eta \mkern1mu x)_2$ by definition of~$\eta$, thus
    $(\pws f) S^X_\Phi \in (\theta f x)_2$ since $f$~is a coalgebra
    homomorphism, and therefore $f x \in \closurename{Y} ((\pws
    f)S^X_\Phi)$. 
    Now observe that whenever $y \in (\pws f) S^X_\Phi$, we have that
    $y = f \mkern-1mu x$ and $\model, x \models \Phi$ for
    some~$x$. Therefore, by inductive hypothesis, $\model',y \models
    \Phi$. In other words, $(\pws f) S^X\Phi \subseteq S^Y_\Phi$. By
    properties of closure, we have $\closure Y {((\pws f)S^X\Phi)}
    \subseteq \closure Y {S^Y_\Phi}$. Thus, by the above derivation,
    we have $f \mkern-1mu x \in \closure Y {S^Y_\Phi}$, that is
    $\model', f x \models \lnear \Phi$.

    ($\Leftarrow$) If $\model', f x \models \lnear \Phi$, then $f x
    \in \closure{Y} S^Y_\Phi$ by definition of~$S^Y_\Phi$, hence
    $S^Y_\Phi \in (\theta f x)_2$ by definition of~$\theta$, and
    $S^Y_\Phi \in ((\closurefunctor f) (\eta \mkern1mu x))_2$ since $f$~is a
    coalgebra homomorphism. Thus $(\pws f) A = S^Y_\Phi$ for some $A
    \in (\eta \mkern1mu x)_2$, hence $(\pws f) A = S^Y_\Phi$ and $x
    \in \closure X A$, from which it follows that $\model',f x'
    \models \Phi$ for all $x' \in A$. By induction hypothesis,
    $\model,x' \models \Phi$ for all $x' \in A$, hence $A \subseteq
    S^X_\Phi$ and $\closure X \mkern-2mu A \subseteq \closure X
    \mkern-1mu S^X_\Phi$ by 
    monotonicity of closure. It follows that $x \in \closure X \mkern-1mu
    S^X_\Phi$ and $\model, x \models \lnear \Phi$, as was to be shown.
    \qed
\end{proof}

%% file: Discussion.tex
\section{Concluding Remarks}
\label{sec:discussion}

In the context of spatial logics and model checking for closure spaces, we have
developed a coalgebraic definition of spatial bisimilarity, a minimization
algorithm, and a free and open source minimisation tool. Bisimilarity
characterises logical equivalence of a finitary logic with two spatial
reachability operators. Furthermore, we have generalised the definition of
topo-bismilarity from topological spaces to closure spaces, proving that the
more general definition still behaves as topo-bisimilarity, in that it
characterises equivalence of infinitary modal logic. Finally, we have provided a
coalgebraic characterisation in the more general setting. Indeed, one of the
primary motivations for our work is the expectation that the tool can be
refined, and the implementation can be integrated with the state-of-the-art
spatial model checker \voxlogica{}, to improve its efficiency, especially when
spatial structures are procedurally generated (e.g.\ by a graph rewriting
procedure or by a process calculus). However, we can identify a number of
theoretical questions, that have the potential to lead to interesting
developments of the research line of spatial model checking.

One major issue that has not yet been addressed is a treatment of logics with
reachability, in the more general setting of Section~\ref{sec:generalization}.
One major difficulty here is that the notion of a \emph{path} has not been
defined in the literature for closure spaces; in \cite{CLLM16} it was emphasized
(see Section 2.4) that the well-known topological definition does not generalise
in the expected way, as it is not compatible with another fundamental notion,
that of paths in a finite graph. Identifying a general notion of path would
allow us to interpret reachability operators in general closure spaces. Such
development is not a merely theoretical exercise. We expect that there are
classes of non-quasi-discrete spaces, that may be finitely represented. For
instance, variants of the polyhedra-based approach of~\cite{BMMP18} may be
relevant for dealing with Euclidean spaces, and in practical terms, for
reasoning about 3D~\emph{meshes} that are of common use in Computer Graphics.
Also spaces that are the union of different components, based either on
polyhedra or on graphs, can give rise to a \emph{hybrid spatial model checking}
approach in the same vein as the celebrated results on model checking of hybrid
systems in the temporal case (see~\cite{Doyen2018}).

Future work should also be devoted to clarifying the generality of the notion of
a closure coalgebra, and to provide a more thorough comparison of closure
coalgebras and neighbourhood frames. In this context, it is also relevant to
investigate the link between closure coalgebras and the treatment of monotone
logics of~\cite{HK04}, given that monotonicity of closure is used in both
directions for the proof of
Theorem~\ref{thm:closureCoalgebraLogicalEquivalence}.

%% file: APX.tex
\section{Appendix: additional proofs}
\label{appendix}


{\bf Lemma~\ref{lemma:aux}}. {\em For $h \in \SET{2,3}$, all $s \in X,
  w \in \Omega$ we have that $w \in (\omega \,
  {\fmorph{\calT}{}{s}})_h$ if and only if $(\eta \, s)_h \cap
  \fmorph{\calT}{-1}{w} \not=\emptyset$.}

\begin{proof}
Since $\fmorph{\calT}{}{\cdot}$ is a $\calT$-homomorphism, we can make the following derivation\\\\
\noindent
$
\deriv
(\omega \, \fmorph{\calT}{}{s})_h
\hint{=}{Def. of $\compose$}
((\omega \compose \fmorph{\calT}{}{\cdot})s)_h
\hint{=}{$\fmorph{\calT}{}{\cdot}$ is a $\calT$-homomorphism}
(((\calT \, \fmorph{\calT}{}{\cdot}) \compose \eta) s)_h
\hint{=}{Def. of $\compose$}
((\calT \, \fmorph{\calT}{}{\cdot})(\eta \, s))_h
\hint{=}{Def. of $(\calT\,\fmorph{\calT}{}{\cdot})$}
\fmorph{\calT}{}{(\eta\, s)_h}
$\\\\
So, $w \in (\omega \, \fmorph{\calT}{}{s})_h$ if and only if $w\in \fmorph{\calT}{}{(\eta\, s)_h}$.
But $w\in \fmorph{\calT}{}{(\eta\, s)_h}$ if and only if there exists $s' \in (\eta\, s)_h$ such that $w = \fmorph{\calT}{}{s'}$, i.e. if and only if $s' \in \fmorph{\calT}{-1}{w}$.
So, $w \in (\omega \, \fmorph{\calT}{}{s})_h$ if and only there exists
$s' \in (\eta\, s)_h  \cap \fmorph{\calT}{-1}{w}$. \\
This proves the assert.
\qed
\end{proof}


\noindent
{\bf Lemma~\ref{lemma:slcseqsubcsbisim}}. {\em If quasi-discrete closure model $\model$ is finite-closure and back-finite-closure, then $\slcseq \, \subseteq\,  \altcmbisim$.}

\begin{proof}
We prove that the equivalence relation $\slcseq$ is a bisimulation  by showing that for all 
$(x_1,x_2) \in \slcseq$ the five conditions of Definition~\ref{def:CMB} are satisfied\footnote{Note 
that Definition~\ref{def:CMB} is used for defining $\cmbisim$, but recall that 
$\altcmbisim$ coincides with $\cmbisim$.}:
\begin{enumerate}
\item $(x_1,x_2) \in \slcseq$ implies $\model, x_1 \models p$ if and only if $\model, x_2 \models p$ for 
all $p \in AP$, which implies in turn that $\invpeval x_1 = \invpeval x_2 $;
\item suppose there exists $x_1' \in \forwclosure{x_1}$ such that $(x_1',x_2') \not\in\slcseq$ for all $x_2' \in \forwclosure{x_2}$. Note that  $x_1' \not=x_1$ because $x_2 \in 
\forwclosure{x_2}\not=\emptyset$ and $(x_1,x_2) \in \slcseq$; moreover
$\forwclosure{x_2}$ is finite since $\model$ is finite-closure. Let then 
$\forwclosure{x_2}=\SET{y_1, \ldots , y_n}$, with $(x_1',y_i) \not\in\slcseq$, for $i=1\ldots n$.
This implies that there would exist formulas $\form_1, \ldots , \form_n$ such that 
$\model, x_1' \models \form_i$ and $\model, y_i \not\models \form_i$, for $i=1\ldots n$, by definition of $\slcseq$.
Thus we would have $\model, x_1' \models \bigwedge_{j=1}^{n} \form_j$ and
$\model, y_i \not\models \bigwedge_{j=1}^{n} \form_j$ for $i=1\ldots n$, which would imply
$\model, x_1 \models \lrcs{(\bigwedge_{j=1}^{n} \form_j)}{\lfalse}$ and
$\model, x_2 \not\models \lrcs{(\bigwedge_{j=1}^{n} \form_j)}{\lfalse}$, and this would 
contradict $(x_1,x_2) \in\slcseq$. Thus we get that for all $x_1' \in \forwclosure{x_1}$ there exists
$x_2' \in \forwclosure{x_2}$ such that $(x_1',x_2')\in \slcseq$;
\item symmetric to the case above;
\item suppose  there exists $x_1' \in \backclosure{x_1}$ such that $(x_1',x_2') \not\in\slcseq$ for all $x_2' \in \backclosure{x_2}$. Note that  $x_1' \not=x_1$ because $x_2 \in 
\backclosure{x_2}\not=\emptyset$ and $(x_1,x_2) \in \slcseq$; moreover
$\backclosure{x_2}$ is finite since $\model$ is finite-back-closure. Let then 
$\backclosure{x_2}=\SET{y_1, \ldots , y_n}$, with $(x_1',y_i) \not\in\slcseq$, for $i=1\ldots n$.
This implies that there would exist formulas $\form_1, \ldots , \form_n$ such that 
$\model, x_1' \models \form_i$ and $\model, y_i \not\models \form_i$, for $i=1\ldots n$, by definition of $\slcseq$.
Thus we would have $\model, x_1' \models \bigwedge_{j=1}^{n} \form_j$ and
$\model, y_i \not\models \bigwedge_{j=1}^{n} \form_j$ for $i=1\ldots n$, which would imply
$\model, x_1 \models \lrcd{(\bigwedge_{j=1}^{n} \form_j)}{\lfalse}$ and
$\model, x_2 \not\models \lrcd{(\bigwedge_{j=1}^{n} \form_j)}{\lfalse}$, and this would 
contradict $(x_1,x_2) \in\slcseq$. Thus we get that for all $x_1' \in \backclosure{x_1}$ there exists
$x_2' \in \backclosure{x_2}$ such that $(x_1',x_2')\in \slcseq$;
\item symmetric to the case above. \qed
\end{enumerate}
\end{proof}


\noindent
{\bf Lemma~\ref{lemma:csbisimsubslcseq}}. {\em 
If $\model$ is a quasi-discrete closure model, then $ \altcmbisim\, \subseteq\,  \slcseq$.}

\medskip
\noindent The proof requires the following lemma.

\begin{lemma}\label{lemma:CFBB}
    For all quasi-discrete models $\model = ((X,\closurename{R}), \peval)$,
    formulas $\Phi$ and $\Psi$, and $x,x' \in X$ the following holds:
    \begin{enumerate}
    \item 
    if $x' \in \forwclosure{x}$ and $\model,x' \models \Phi$ then $\model, x \models \lrcs{\Phi}{\Psi}$;
    \item
    if $x' \in \backclosure{x}$ and $\model,x' \models \Phi$ then $\model, x \models \lrcd{\Phi}{\Psi}$.
    \end{enumerate}
\end{lemma}
    
\begin{proof}(of Lemma~\ref{lemma:CFBB})
    Keeping in mind that $Y \subseteq \closure{}{Y}$ for all $Y \subseteq X$
    \begin{enumerate}
    \item
    take  $\pi: \nats \to X$ with $\pi(0) = x$ and $\pi(j) = x'$ for all $j\in \nats, j>0$;
    $\pi$ is a path since for all $N \subseteq \nats$ we have
    $$
    \pi(\closure{Succ}{N})=
    \left\{
    \begin{array}{l}
    \emptyset, \mbox{ if } N=\emptyset,\\
    \SET{x'}, \mbox{ if } 0\not\in N\not=\emptyset,\\
    \SET{x,x'}, \mbox{ if } 0\in N.
    \end{array}
    \right.
    \quad
    \closure{R}{\pi(N)}=
    \left\{
    \begin{array}{l}
    \emptyset, \mbox{ if } N=\emptyset,\\
    \closure{R}{\SET{x'}}, \mbox{ if } 0\not\in N\not=\emptyset,\\
    \closure{R}{\SET{x,x'}}, \mbox{ if } 0\in N.\\
    \end{array}
    \right.
    $$
    so that $\pi(\closure{Succ}{N}) \subseteq \closure{R}{(\pi N)}$;
    \item
    note that if $x' \in \backclosure{x}$ then $x \in \forwclosure{x'}$ and 
    take  $\pi: \nats \to X$ with $\pi(0) = x'$ and $\pi(j) = x$ for all $j\in \nats, j>0$;
    $\pi$ is a path since for all $N \subseteq \nats$ we have
    $$
    \pi(\closure{Succ}{N})=
    \left\{
    \begin{array}{l}
    \emptyset, \mbox{ if } N=\emptyset,\\
    \SET{x}, \mbox{ if } 0\not\in N\not=\emptyset,\\
    \SET{x,x'}, \mbox{ if } 0\in N.
    \end{array}
    \right.
    \quad
    \closure{R}{\pi(N)}=
    \left\{
    \begin{array}{l}
    \emptyset, \mbox{ if } N=\emptyset,\\
    \closure{R}{\SET{x}}, \mbox{ if } 0\not\in N\not=\emptyset,\\
    \closure{R}{\SET{x,x'}}, \mbox{ if } 0\in N.
    \end{array}
    \right.
    $$
    so that $\pi(\closure{Succ}{N}) \subseteq \closure{R}{(\pi N)}$.\qed
    \end{enumerate}
\end{proof}

\begin{proof}(of Lemma~\ref{lemma:csbisimsubslcseq})
By induction on the structure of $\form$ we prove that, for all
$x_1, x_2 \in X$ and for all $ \SLCS$  formulas $\form$ ,  if $(x_1,x_2) \in \altcmbisim$, then
$\model, x_1 \models \form$ if and only if $\model, x_2 \models \form$.

\noindent
{\bf Base case} $p$:\\
$(x_1,x_2) \in \altcmbisim$ implies $\invpeval x_1 =\invpeval x_2 $ which implies in turn 
$\model, x_1 \models p$ if and only if $\model, x_2 \models p$, for all $p\in AP$.\\

\noindent
{\bf Induction steps}\\
We assume the induction hypothesis---for all $x_1,x_2 \in X$, if $(x_1,x_2) \in \altcmbisim$, then the following holds
$\model, x_1 \models \form$ if and only if $\model, x_2 \models \form$ for any $\SLCS$  formula $\Phi$---and we prove the following cases, for any $(x_1,x_2) \in \altcmbisim$:\\
{\bf Case} $\lneg\form$:\\
Suppose $\model, x_1 \models \lneg\form$ and $\model, x_2 \not\models \lneg\form$.
This would imply $\model, x_1 \not\models \form$ and $\model, x_2 \models \form$ and since
$(x_1,x_2) \in \altcmbisim$, this would contradict the induction hypothesis. 

\noindent
{\bf Case} $\Phi\, \land \Psi$:\\
Suppose $\model, x_1 \models \Phi\, \land \Psi$ and $\model, x_2 \not\models \Phi\, \land \Psi$
and w.l.g. assume $\model, x_2 \not\models \Phi$. 
Then we would get $\model, x_1 \models \Phi$ and $\model, x_2 \not\models \Phi$ and since
$(x_1,x_2) \in \altcmbisim$, this would contradict the induction hypothesis. 

\noindent
{\bf Case} $\lrcs{\Phi}{\Psi}$:\\
Suppose $\model, x_1 \models \lrcs{\Phi}{\Psi}$ and $\model, x_2 \not\models \lrcs{\Phi}{\Psi}$.
$\model, x_1 \models \lrcs{\Phi}{\Psi}$ means there exists path $\pi_1$ and index $\ell$ such that
$\pi_1(0)=x_1, \model,\pi_1(\ell) \models \Phi$ and $\model,\pi_1(j) \models \Psi$ for all 
$j \in I_{\ell}$ where we define $I_n$ as $I_n = \SET{1,\ldots n-1}$. We distinguish three cases:
\begin{itemize}
\item $\ell=0$: in this case, by definition of $\lrcs{\Phi}{\Psi}$, $\model, x_1 \models \Phi$; 
on the other hand, since  $\model, x_2 \not\models \lrcs{\Phi}{\Psi}$
by hypothesis, it should hold $\model, x_2 \not\models \Phi$, but since
$(x_1,x_2) \in \altcmbisim$, this would contradict the induction hypothesis;
\item  $\ell=1$: in this case $\model, \pi_1(1) \models \Phi$ and, by continuity of $\pi_1$, 
we would have that $\pi_1(1)\in\closure{R}{\SET{x_1}}$;
in fact, continuity of $\pi_1$ implies $\pi_1(\closure{R}{\SET{0}}) \subseteq
\closure{\succ}{(\pi_1(\SET{0}))}$, so that we get the following derivation:
$
\pi_1(1) \in \SET{\pi_1(0),\pi_1(1)} \\ = \pi_1(\SET{0,1}) =
\pi_1(\closure{\succ}{\SET{0}}) \subseteq
\closure{R}{(\pi_1(\SET{0}))} =
\closure{R}{\SET{\pi_1(0)}} =
\closure{R}{\SET{x_1}}
$;
that is, there exists $x_1' \in\forwclosure{x_1}$  such that 
$\model, x_1' \models \Phi$; on the other hand, since  $\model, x_2 \not\models \lrcs{\Phi}{\Psi}$
by hypothesis, it should hold $\model, x_2' \not\models \Phi$ for all $x_2' \in\forwclosure{x_2}$,
due to Lemma~\ref{lemma:CFBB}(1) below;
moreover, we know that $(x_1,x_2) \in \altcmbisim$, which, by definition of $\altcmbisim$, and recalling that 
$\altcmbisim$ coincides with $\cmbisim$, implies that there would exist $x_2' \in\forwclosure{x_2}$
such that $(x_1',x_2') \in \altcmbisim$; but then we would have
$\model, x_1' \models \Phi$ and $\model, x_2' \not\models \Phi$ that contradicts the induction hypothesis;
\item $\ell>1$: in this case we can build a path $\pi_2$ as follows: $\pi_2(0)=x_2$, 
{$\pi_2(j) \in \; \forwclosure{\pi_2(j-1)}$} for $j\in I_{\ell}$, and $(\pi_1(j),\pi_2(j))\in \altcmbisim$ for 
$j=0, \ldots \ell-1$; in fact $(\pi_1(0),\pi_2(0))\in \altcmbisim$ by hypothesis and this
implies there exists $x_2' \in \forwclosure{\pi_2(0)}$ such that $(\pi_1(1),x_2')\in \altcmbisim$ and we 
let $\pi_2(1)=x_2'$; a similar reasoning can now be applied starting from 
$(\pi_1(1),\pi_2(1))\in \altcmbisim$, $(\pi_1(2),\pi_2(2))\in \altcmbisim$
and so on till $(\pi_1(\ell-1),\pi_2(\ell-1))\in \altcmbisim$; since $\model,\pi_1(j) \models \Psi$ for
all $j\in I_{\ell}$, by the induction hypothesis we get that also $\model,\pi_2(j) \models \Psi$ for
all $j\in I_{\ell}$; note moreover that  $\model,\pi_2(j) \not\models \Phi$ for all $j=0,\ldots,\ell-1$ since, by hypothesis,
$\model, x_2 \not\models \lrcs{\Phi}{\Psi}$ and, for the same reason, it should also be the case that
$\model,z \not\models \Phi$ for all $z \in \forwclosure{\pi_2(\ell-1)}$; and since 
$(\pi_1(\ell-1),\pi_2(\ell-1))\in \altcmbisim$ there should be a $z \in \forwclosure{\pi_2(\ell-1)}$
such that  $(\pi_1(\ell), z) \in \altcmbisim$; but then, the induction hypothesis would be violated
by $\model, \pi_1(\ell) \models \Phi$ and $\model, z \not\models \Phi$.
\end{itemize}

\noindent
{\bf Case} $\lrcd{\Phi}{\Psi}$:\\
Suppose $\model, x_1 \models \lrcd{\Phi}{\Psi}$ and $\model, x_2 \not\models \lrcd{\Phi}{\Psi}$.
$\model, x_1 \models \lrcd{\Phi}{\Psi}$ means there exists path $\pi_1$ and index $\ell$ such that
$\pi_1(\ell)=x_1, \model,\pi_1(0) \models \Phi$ and $\model,\pi_1(j) \models \Psi$ for all 
$j \in I_{\ell}$. We distinguish three cases:
\begin{itemize}
\item $\ell=0$: in this case, by definition of $\lrcd{\Phi}{\Psi}$, $\model, x_1 \models \Phi$;
on the other hand, since  $\model, x_2 \not\models \lrcd{\Phi}{\Psi}$
by hypothesis, it should hold that $\model, x_2 \not\models \Phi$, but since
$(x_1,x_2) \in \altcmbisim$, this would contradict the induction hypothesis;
\item  $\ell=1$: in this case we have $\model, \pi_1(0) \models \Phi$ and  $\pi_1(1) = x_1$.
We first note that, by continuity of $\pi_1$, we have $x_1 \in \closure{R}{\SET{\pi_1(0)}}$:
$
x_1 =
\pi_1(1) \in
\SET{\pi_1(0),\pi_1(1)} =
\pi_1(\SET{0,1})=
\pi_1(\closure{\succ}{\SET{0}}) \subseteq
\closure{R}{\pi_1(\SET{0})}=
\closure{R}{(\SET{\pi_1(0)})}
$.
This means that there exists $x_1' \in \backclosure{x_1}$ such that $\model, x_1' \models \Phi$, namely $x_1'=\pi_1(0)$. We also know that $\model, x_2' \not\models \Phi$ for all
 $x_2' \in \backclosure{x_2}$, otherwise, by Lemma~\ref{lemma:CFBB}(2) below, $\model, x_2 \models \lrcd{\Phi}{\Psi}$ would hold,
 which is not the case by hypothesis. On the other hand, again by hypothesis we know that
$(x_1,x_2) \in \altcmbisim$, and so, given that $x_1' \in \backclosure{x_1}$, there must also be 
some $x_2'' \in \backclosure{x_2}$ such that $(x_1',x_2'') \in \altcmbisim$. But this, 
by the induction hypothesis, implies that $\model, x_2'' \models \Phi$ which 
contradicts the fact that $\model, x_2' \not\models \Phi$ for all
 $x_2' \in \backclosure{x_2}$.
\item $\ell>1$: in this case we can build a path $\pi_2$ as follows: $\pi_2(\ell)=x_2$, 
$\quad \pi_2(j) \; \in \; \forwclosure{\pi_2(j-1)}$ for $j\in I_{\ell}$, and $(\pi_1(j),\pi_2(j))\in \altcmbisim$ for 
$j=0, \ldots \ell-1$; in fact $(\pi_1(\ell),\pi_2(\ell))\in \altcmbisim$ by hypothesis and this
implies there exists $x_2' \in \backclosure{x_2}$ such that $(\pi_1(\ell-1),x_2')\in \altcmbisim$,
because $\pi_1(\ell -1)\in \backclosure{x_1}$.
We let $\pi_2(\ell-1)=x_2'$; a similar reasoning can now be applied starting from 
$(\pi_1(\ell-1),\pi_2(\ell-1))\in \altcmbisim$, 
and so on till $(\pi_1(1),\pi_2(1))\in \altcmbisim$. 
Since $\model,\pi_1(j) \models \Psi$ for all $j\in I_{\ell}$, by the induction 
hypothesis we get that also $\model,\pi_2(j) \models \Psi$ for
all $j\in I_{\ell}$; note moreover that  $\model,\pi_2(j) \not\models \Phi$ for all $j=0,\ldots,\ell-1$, because 
$\model, x_2 \not\models \lrcd{\Phi}{\Psi}$ and, for the same reason, it should also be the case that
$\model,z \not\models \Phi$ for all $z \in \backclosure{\SET{\pi_2(1)}}$.
But we know that $(\pi_1(1),\pi_2(1)) \in \altcmbisim$ and that 
$\pi_1(0) \in\backclosure{\pi_1(1)}$, so there should be $z \in \backclosure{\pi_2(1)}$
and $(\pi_1(0),z) \in \altcmbisim$ and that $\model, \pi_1(0)\models \Phi$. For the
induction hypothesis, then, we should have also $\model, z \models \Phi$, which brings
to contradiction.
\end{itemize}
This completes the proof.
\end{proof}


\noindent
{\bf Lemma~\ref{lemma:restrict}}. {\em 
For $\form_1,\form_2$ formulas of $\SMCSm$ of Equation~\ref{eq:sublogic}, 
\begin{enumerate}
\item
if $\model, x_1 \models \lrcs{\form_1}{\form_2}$ 
and $\model, x_2 \not\models \lrcs{\form_1}{\form_2}$
then there exists $\Lambda_{\form_1,\form_2}$ 
in the language of Equation~\ref{eq:sublogic} such that
$\model, x_1 \models \Lambda_{\form_1,\form_2} $ 
and $\model, x_2 \not\models \Lambda_{\form_1,\form_2} $;
\item
if $\model, x_1 \models \lrcd{\form_1}{\form_2}$ 
and $\model, x_2 \not\models \lrcd{\form_1}{\form_2}$
then there exists $\Lambda_{\form_1,\form_2}$
in the language of Equation~\ref{eq:sublogic} such that
$\model, x_1 \models \Lambda_{\form_1,\form_2} $ 
and $\model, x_2 \not\models \Lambda_{\form_1,\form_2} $.
\end{enumerate}
}

\begin{proof}
$ $\\
For what concerns item (1),
there are three cases for $\model, x_1 \models \lrcs{\form_1}{\form_2}$.\\
{\bf Case 1:} $\ell = 0$.\\
By definition of the $\lrcsname$ operator, in this case we have $\model, x_1 \models \form_1$.
On the other hand, since $\model, x_2 \not\models \lrcs{\form_1}{\form_2}$, 
we have $\model, x_2 \not\models \form_1$, otherwise 
$\model, x_2 \models \lrcs{\form_1}{\form_2}$ would hold, by definition of $\lrcsname$.
So, in this case  $\Lambda_{\form_1,\form_2} = \form_1$.\\
{\bf Case 2:} $\ell = 1$.\\
By definition of the $\lrcsname$ operator, in this case we have that 
there exists a path $\pi_1$ such that $\pi_1(0)=x_1$ and
$\model, \pi_1(1) \models \form_1$. 
This means that 
$\model, x_1 \models \lrcs{\form_1}{\lfalse}$. On the other hand, 
from the fact that $\model, x_2 \not\models \lrcs{\form_1}{\form_2}$ we get, 
again by definition of  $\lrcsname$,
$\model, x_2 \not\models \lrcs{\form_1}{\lfalse}$. 
So, in this case, $\Lambda_{\form_1,\form_2} = \lrcs{\form_1}{\lfalse}$.
\\
{\bf Case 3:} $\ell = k>1$.\\
By definition of the $\lrcsname$ operator, in this case we have that 
there exists a path $\pi_1$ such that $\pi_1(0)=x_1$,
$\model, \pi_1(k) \models \form_1$ and $\model, \pi_1(j) \models \form_2$, for 
$0<j<k$. 

\noindent
It  is easy to see that:\\
$\model, \pi_1(k-1) \models \form_2 \, \land \, 
\lrcs{\form_1}{\lfalse}$\\
$\model, \pi_1(k-2) \models \form_2 \, \land \, 
\lrcs{(\form_2 \, \land \, \lrcs{\form_1}{\lfalse})}{\lfalse}$\\
\vdots\\
$\model, x_1  \models \Psi$ 
where
$\Psi =
\underbrace{
\lrcs
{(\form_2 \, \land \,
\lrcs{(\ldots \, \land \, \lrcs{(\form_2 \, \land \, \lrcs{\form_1}{\lfalse})}{\lfalse} \ldots}{\lfalse})
)
}
{\lfalse}
}_{k \mbox{ times } \lrcsname}
$,\\
and that:
$\model, x_2  \models \Psi$
does not hold, otherwise one could easily build a
path $\pi_2$ with $\pi_2(0)=x_2$,
$\model, \pi_2(k) \models \form_1$ and $\model, \pi_2(j) \models \form_2$, for 
$0<j<k$ and, consequently we would have $\model, x_2 \models \lrcs{\form_1}{\form_2}$.
So, in this case, $\Lambda_{\form_1,\form_2} = \Psi$.

\noindent
For what concerns point (2),
there are three cases for $\model, x_1 \models \lrcd{\form_1}{\form_2}$.\\
{\bf Case 1:} $\ell = 0$.\\
By definition of the $\lrcdname$ operator, in this case we have $\model, x_1 \models \form_1$.
On the other hand, since $\model, x_2 \not\models \lrcd{\form_1}{\form_2}$, 
we have $\model, x_2 \not\models \form_1$, otherwise 
$\model, x_2 \models \lrcd{\form_1}{\form_2}$ would hold, by definition of $\lrcdname$.
So, in this case  $\Lambda_{\form_1,\form_2} = \form_1$.\\
{\bf Case 2:} $\ell = 1$.\\
By definition of the $\lrcdname$ operator, in this case we have that 
there exists a path $\pi_1$ such that $\model, \pi_1(0) \models \form_1$ and
$\pi_1(1)=x_1$. 
This means that 
$\model, x_1 \models \lrcd{\form_1}{\lfalse}$. On the other hand, 
from the fact that $\model, x_2 \not\models \lrcd{\form_1}{\form_2}$ we get, 
again by definition of  $\lrcdname$,
$\model, x_2 \not\models \lrcd{\form_1}{\lfalse}$. 
So, in this case, $\Lambda_{\form_1,\form_2} = \lrcd{\form_1}{\lfalse}$.
\\
{\bf Case 3:} $\ell = k>1$.\\
By definition of the $\lrcdname$ operator, in this case we have that 
there exists a path $\pi_1$ such that $\model, \pi_1(0) \models \form_1$, 
$\pi_1(k)=x_1$, and $\model, \pi_1(j) \models \form_2$, for 
$0<j<k$. 

\noindent
It  is easy to see that:\\
$\model, \pi_1(1) \models \form_2 \, \land \, 
\lrcd{\form_1}{\lfalse}$\\
$\model, \pi_1(2) \models \form_2 \, \land \, 
\lrcd{(\form_2 \, \land \, \lrcd{\form_1}{\lfalse})}{\lfalse}$\\
\vdots\\
$\model, x_1  \models \Gamma$
where
$\Gamma =
\underbrace{
\lrcd
{(\form_2 \, \land \,
\lrcd{(\ldots \, \land \, \lrcd{(\form_2 \, \land \, \lrcd{\form_1}{\lfalse})}{\lfalse} \ldots}{\lfalse})
)
}
{\lfalse}
}_{k \mbox{ times } \lrcdname}
$,\\
and that:\\ 
$\model, x_2  \models \Gamma$
does not hold, otherwise one could easily build a
path $\pi_2$ with $\model, \pi_2(0) \models \form_1$,  $\pi_2(k)=x_2$,
and $\model, \pi_2(j) \models \form_2$, for 
$0<j<k$ and, consequently we would have $\model, x_2 \models \lrcd{\form_1}{\form_2}$.
So, in this case $\Lambda_{\form_1,\form_2} = \Gamma$.
\qed
\end{proof}


\noindent{\bf Lemma~\ref{lem:quotientsOfClosureCoalgebras}.} 
    Consider $((X,\closurename{X}),\pevalname)$ and $\calX = (X,\eta)$ as in
    Definition \ref{def:closureCoalgebraOfAModel}. Let $\calY = (Y,\theta)$ be a
    $\closurefunctor$-coalgebra. Let $f : \calX \twoheadrightarrow \calY$ be a
    surjective coalgebra homomorphism. Define $\closure{Y}{(B \subseteq Y)} =
    \{ y \in Y \mid B \in (\theta y)_2 \}$. Then $(Y,\closurename{Y})$ is a
    closure space. 

    \medskip

    \noindent The proof requires the following lemma and its corollary. 

    \begin{lemma}\label{lem:closure-and-image}
        Consider $((X,\closurename{X}),\pevalname)$ and $\calX = (X,\eta)$ as in
    Definition \ref{def:closureCoalgebraOfAModel}. Let $\calY = (Y,\theta)$ be a $\closurefunctor$-coalgebra. Let $f : \calX
    \to \calY$ be a (not necessarily surjective)
    coalgebra homomorphism. Define $\closure{Y}{(B \subseteq Y)} = \{ y \in Y
    \mid B \in (\theta y)_2 \}$. It holds that $\forall A \subseteq X . \forall
    x \in X . x \in \closure{X}{A} \iff f x \in \closure{Y}{(\pws f)A}$, that
    is, $\closure{X}{A} = f^{-1}(\closure{Y}{(\pws f)A})$.
    \end{lemma}

    \begin{corollary}\label{cor:closure-and-epi-image}
        Under the conditions of Lemma~\ref{lem:quotientsOfClosureCoalgebras}, that is, whenever Lemma~\ref{lem:closure-and-image} holds, and $f$ is surjective, for $B \subseteq Y$, we have $\closure{Y}{B} = (\pws f) \closure{X}{f^{-1}B}$.
    \end{corollary}

    \begin{proof}(of Lemma~\ref{lem:closure-and-image})
    $\nh{}
    x \in \closure{X}{A} \iff A \in (\eta x)_2 \iff (\pws f) A \in ((\closurefunctor f) \eta x)_2 \hint{\iff}{\text{$f$ is a coalgebra homomorphism}} (\pws f) A \in (\theta f x)_2 \iff f x \in \closure{Y}{(\pws f) A}
    $\qed
    \end{proof}

\begin{proof}(of Lemma~\ref{lem:quotientsOfClosureCoalgebras})    
    
    \medskip
    \noindent {\bf If $\closure{X}{\emptyset} = \emptyset$ holds}, we have:
    $\nh{}
    \closure{Y}{\emptyset} 
    \hint = {\text{Corollary~\ref{cor:closure-and-epi-image}}}
    (\pws f)(\closure{X}{\emptyset})
    \hint{=}{\text{Hypothesis on $\closurename{X}$}}
    \pws f \emptyset
    \nh = 
    \emptyset
    $\qed

    \medskip \noindent {\bf If $\forall A \subseteq X . A \subseteq \closure{X}{A}$ holds}, for $B \subseteq Y$, by the hypothesis $f^{-1}B \subseteq \closure{X}{f^{-1} B}$, and $f$ being surjective, $B \subseteq (\pws f) \closure{X}{f^{-1} B}$, and by Corollary~\ref{cor:closure-and-epi-image}, $B \subseteq \closure{Y}{B}$. \qed

    \medskip \noindent {\bf If $\forall A,B \subseteq X . (\closure{X}{A})\cup (\closure{X}{B}) = \closure{X}{(A \cup B)}$ holds}, for $C,D \subseteq Y$, we have 
    $\nh{}
    \closure{Y}{(C \cup D)} 
    \hint = {\text{Corollary~\ref{cor:closure-and-epi-image}}} 
    (\pws f)(\closure{X}{f^{-1}(C \cup D)}) 
    \nh = 
    (\pws f)(\closure{X}{((f^{-1}C) \cup (f^{-1}D))})
    \hint = {\text{Hypothesis on $\closurename{X}$}}
    (\pws f)((\closure{X}{f^{-1}C}) \cup  (\closure{X}{f^{-1}D}))
    \nh =  
    ((\pws f)(\closure{X}{f^{-1}C})) \cup  ((\pws f)(\closure{X}{f^{-1}D}))
    \hint = {\text{Corollary~\ref{cor:closure-and-epi-image}}} 
    (\closure{Y}{C}) \cup (\closure{Y}{D})
    $\qed
\end{proof}

\noindent \textbf{Lemma \ref{lem:closure-of-equivalent-points}.}
    Let $f$ be the function mapping each element of $X$ into its equivalence classs up to $\inflogeq$. For all $x_1,x_2 \in X$ and $A \subseteq X$,  it holds that $((x_1 \inflogeq x_2) \land x_1 \in \closure{X} A) \implies x_2 \in \closurename{X} f^{-1} (\pws f) A$.
\medskip

\begin{proof}
    For any $x \in X$ let $\chi_x$ be a formula that holds on any $x'$ if and only if $x' \inflogeq x$; such a formula is the (possibly infinite) conjunction of the formulas telling apart $[x]$ from the other equivalence classes of $\inflogeq$. Let $\Sigma = \bigwedge_{z \notin f^{-1} (\pws f) A} \lnot \chi_z$. We have $\model,y \models \Sigma \iff y \in f^{-1} (\pws f) A$. 

    It is true that $A \subseteq f^{-1} (\pws f) A$. Therefore, by properties of closure spaces, we have $\closure{X}{A} \subseteq \closure{X}f^{-1} (\pws f) A$. Thus, by the hypothesis $x_1 \in \closure{X} A$, we have $\model, x_1 \models \lnear \Sigma$. By logical equivalence, also $\model, x_2 \models \lnear \Sigma$. Therefore $x_2 \in \closurename{X}{f^{-1}(\pws f) A}$. \qed
\end{proof}